%% file: Grid.tex
\documentclass[11pt,a4paper,UKenglish]{article}

\usepackage[margin=1in]{geometry}

\usepackage{authblk}

\usepackage{amsmath}
\usepackage{amssymb}
\usepackage{amsthm}
\newtheorem{theorem}{Theorem}
\newtheorem{observation}[theorem]{Observation}

\newtheorem{lemma}[theorem]{Lemma}
\newtheorem{corollary}[theorem]{Corollary}
\newtheorem{example}[theorem]{Example}

\usepackage{graphicx}
\usepackage{algorithm}
\usepackage{algorithmicx}
\usepackage[noend]{algpseudocode}

\usepackage{hyperref}
\hypersetup{pdftex}
\usepackage[usenames,dvipsnames,svgnames,table]{xcolor}
\usepackage{url}
\usepackage{cite}

\usepackage{microtype}

\usepackage{multicol}
\usepackage{multirow}

\usepackage{soul}

\title{%
Non-preemptive Scheduling in a Smart Grid Model and\\ its Implications on Machine  Minimization\thanks{A preliminary version of this paper appeared in Proceedings of the 27th International Symposium on Algorithms and Computation, ISAAC~2016~\cite{LiuLW16} and some results are improved in this version.}
}

\author[1]{Fu-Hong Liu}
\author[1,2]{Hsiang-Hsuan Liu}
\author[2]{Prudence W.H. Wong}
\affil[1]{Department of Computer Science, National Tsing Hua University\\
  101 Kuang Fu Road, Hsinchu, Taiwan\\
  \texttt{\{fhliu,hhliu\}@cs.nthu.edu.tw}}
\affil[2]{Department of Computer Science, University of Liverpool, Liverpool, UK\\
  \texttt{\{hhliu,pwong\}@liverpool.ac.uk}}

\input{preamble}

\begin{document}

\maketitle

\input{abstract}

\input{intro}

\input{prelim}

\input{verticalUniform}

\input{general}

\input{general_lb}

\input{horizontalUniform}


\input{exact}

\input{peak}

\input{conclusion}

\section*{Acknowledgement}
The work is partially supported by Networks Sciences and Technologies, University of Liverpool. Hsiang-Hsuan Liu is partially supported by a studentship from the University of Liverpool-National Tsing-Hua University Dual PhD programme.

\bibliography{ref_sch} 

\end{document}

%% file: preamble.tex
\newcommand{\comment}[1]{}

 \newcommand{\pw}[1]{{\color{red}#1}}
 \newcommand{\hhl}[1]{\textcolor{OliveGreen}{#1}}
 \newcommand{\fhl}[1]{{\color{RoyalBlue}#1}}

\newcommand{\JS}{{\mathcal{J}}}
\newcommand{\Job}[1]{{J_{#1}}}
\newcommand{\Jobset}{{\mathcal{J}}}
\newcommand{\Jobsetij}[1]{{\mathcal{J}_{#1}}}
\newcommand{\JSpq}{{\mathcal{J}_{p,q}}}

\newcommand{\niceJSpq}{{\JS^*_{p,q}}}

\newcommand{\relaxS}{\text{\sc RelaxSch}}
\newcommand{\shrinkS}{\text{\sc ShrinkSch}}
\newcommand{\convert}{\text{\sc Convert}}
\newcommand{\Alg}{{\mathcal{A}}}

\newcommand{\AlgU}{{\mathcal{U}}}
\newcommand{\AlgOne}{{\mathcal{UU}}}
\newcommand{\AlgAgree}{{\mathcal{AD}}}
\newcommand{\AlgE}{{\mathcal{E}}}

\newcommand{\YDS}{{\mathcal{YDS}}}
\newcommand{\Opt}{{\mathcal{O}}}
\newcommand{\AlgV}{{\mathcal{V}}}
\newcommand{\grid}{{\mathsf{GRID}}}

\newcommand{\AlgUV}{{\mathcal{UV}}}
\newcommand{\convertfi}{\text{\sc AlignFI}}
\newcommand{\convertw}{\text{\sc Convert}}
\newcommand{\uwJ}{{J}}
\newcommand{\uwaJ}{{J^\prime}}
\newcommand{\uwJS}{{\JS}}
\newcommand{\uwaJS}{{\JS^\prime}}
\newcommand{\uwJSl}{{\uwJS_{\text{L}}}}
\newcommand{\uwJSt}{{\uwJS_{\text{T}}}}
\newcommand{\uwJSp}{{\JS^*_p}}

\newcommand{\alignS}{\text{\sc AlignSch}}
\newcommand{\freeS}{\text{\sc FreeSch}}
\newcommand{\uwSch}{{\sch}}
\newcommand{\uwaSch}{{\sch^\prime}}
\newcommand{\uwaJSl}{{\JS^\prime_{\text{L}}}}
\newcommand{\JSp}{{\mathcal{J}_p}}

\newcommand{\Width}[1]{{w(\Job{#1})}}
\newcommand{\Height}[1]{{h(\Job{#1})}}

\newcommand{\load}{{\ell}}
\newcommand{\Load}[2]{{\ell(#1,#2)}}
\newcommand{\rel}{{r}}
\newcommand{\dl}{{d}}
\newcommand{\work}{{p}}
\newcommand{\interval}{{I}}
\newcommand{\I}{{\mathcal{I}}}
\newcommand{\IL}[1]{{\mathcal{I}_{>#1}}}
\newcommand{\IS}[1]{{\mathcal{I}_{\leq#1}}}

\newcommand{\w}{{w}}
\newcommand{\h}{{h}}
\newcommand{\hmax}{{\h_{\max}}}

\newcommand{\wmax}{{\w_{\max}}}
\newcommand{\wmin}{{\w_{\min}}}

\newcommand{\stime}{{st}}
\newcommand{\etime}{{et}}
\newcommand{\Kw}{{K}}
\newcommand{\Kh}{{K_h}}

\newcommand{\ceilLogKw}{{\lceil \log \Kw \rceil}}
\newcommand{\logKh}{{k_h}}
\newcommand{\class}{{C}}
\newcommand{\den}{{\text{den}}}
\newcommand{\cost}{{\text{cost}}}
\newcommand{\sch}{{S}}
\newcommand{\schp}{{S_p}}
\newcommand{\schpq}{{S_{p,q}}}

\newcommand{\uwSchp}{{\sch_p^*}}
\newcommand{\avg}{{\text{avg}}}
\newcommand{\AVR}{{\mathcal{AVR}}}

\newcommand{\R}{{\mathcal{R}}}
\newcommand{\BKP}{{\mathcal{BKP}}}
\newcommand{\newBKP}{{\mathcal{BKP}^\prime}}

\newcommand{\runtitle}[1]{{\textbf{\boldmath #1}}}

\newcommand{\queue}{{\mathcal{Q}}}

\newcommand{\spl}{{spl}}
\newcommand{\cls}{{cls}}

\newcommand{\Tl}{{T_{\text{left}}}}
\newcommand{\Tr}{{T_{\text{right}}}}
\newcommand{\win}{{W}}
\newcommand{\bound}{{b}}
\newcommand{\config}{{F}}
\newcommand{\configL}{{F_{\text{left}}}}
\newcommand{\configR}{{F_{\text{right}}}}
\newcommand{\AlgEPlus}{{\mathcal{E}^+}}

\newcommand{\Ja}{{\JS}}
\newcommand{\adv}{{\mathit{\Lambda}}}

\newcommand{\peak}{{\mathsf{GRID_{peak}}}}
\newcommand{\Sch}{{S}}
\newcommand{\dvs}{{\mathsf{DVS}}}

\newcommand{\p}{{\text{peak}}}

\newcommand{\niceJobsetij}[1]{{\Jobset^*_{#1}}}

\newcommand{\niceWidth}[1]{{w(J^*_{#1})}}
\newcommand{\niceHeight}[1]{{h(J^*_{#1})}}

\newcommand{\niceRtime}[1]{{r(J^*_{#1})}}
\newcommand{\niceDline}[1]{{d(J^*_{#1})}}

\newcommand{\Schij}[1]{{S_{#1}}}

\newcommand{\alignSchij}[1]{{S'_{#1}}}

\newcommand{\JStight}{{\Jobset_{\text{T}}}}

\newcommand{\niceJStight}{{\Jobset^*_{\text{T}}}}
\newcommand{\niceJSloose}{{\Jobset^*_{\text{L}}}}

\newcommand{\alignJobsetij}[1]{{\Jobset'_{#1}}}
\newcommand{\alignJob}[1]{{J'_{#1}}}
\newcommand{\alignWidth}[1]{{w(J'_{#1})}}
\newcommand{\alignHeight}[1]{{h(J'_{#1})}}

\newcommand{\alignRtime}[1]{{r(J'_{#1})}}
\newcommand{\alignDline}[1]{{d(J'_{#1})}}

\newcommand{\Rtime}[1]{{r(\Job{#1})}}
\newcommand{\Dline}[1]{{d(\Job{#1})}}
\newcommand{\Work}[1]{{p(\Job{#1})}}

\newcommand{\minimization}{{\mathsf{MACHINE}}}

%% file: abstract.tex
\begin{abstract}
We study a scheduling problem arising in demand response management in smart grid.
Consumers send in power requests with a flexible feasible time interval
during which their requests can be served.
The grid controller, upon receiving power requests,
schedules each request within the specified interval.
The electricity cost is measured by a convex function
of the load in each timeslot.
The objective is to schedule all requests
with the minimum total electricity cost. 
Previous work has studied cases where jobs have unit power requirement
and unit duration. 
We extend the study to arbitrary power requirement and duration,
which has been shown to be NP-hard.
We give the first online 
algorithm for the general problem, 
and prove that the problem is fixed parameter tractable.
We also show that the online algorithm is asymptotically optimal when the objective is to minimize the peak load. 
In addition, we observe that the classical non-preemptive machine minimization problem 
is a special case of the smart grid problem with min-peak objective, 
and show that we can solve the non-preemptive machine minimization problem asymptotically optimally.
\comment{We also show that the online algorithm is asymptotically optimal with respect to the smart grid problem with minimizing peak objective. Since the non-preemptive machine minimization problem is a special case of the min peak smart grid problem, our algorithm can solve the non-preemptive machine minimization problem optimally.}
\end{abstract}

%% file: intro.tex
\section{Introduction}
\label{sec:intro}

We study a scheduling problem arising in
``demand response management'' in smart grid~\cite{HG10,IA09,LSM10,IEEE12,FMX+12}.
The electrical smart grid is one of the major challenges
in the 21st century~\cite{DOE09,ESG06,UKDEC13}.
The smart grid \cite{farhangi2010path,masters2013renewable}  is a power grid system that makes power generation, distribution and consumption more efficient through information and communication technologies against the traditional power system.
Peak demand hours happen only for a short duration,
yet makes existing electrical grid less efficient.
It has been noted in~\cite{CNX+13} that
in the US power grid, 10\% of all generation assets and
25\% of distribution infrastructure are required for less than 400
hours per year, roughly 5\% of the time~\cite{DOE09}.
\emph{Demand response management}
attempts to overcome this problem by shifting users'
demand to off-peak hours in order to
reduce peak load~\cite{LSS12,MZZ+13,KT11,CK10,SLL13,mohsenian2010autonomous}.
Research initiatives in the area include~\cite{kannberg2003gridwise,M13,RE13,PS13}.

The electricity grids supports demand response mechanism 
and obtains energy efficiency by organizing customer consumption of electricity 
in response to supply conditions. 
It is demonstrated in \cite{LSM10} that demand response is of remarkable advantage to consumers, utilities, and society. 
Effective demand load management brings down the cost of operating the grid, 
as well as energy generation and distribution~\cite{LSS12}.
Demand response management is not only advantageous to the supplier but also to
the consumers as well.
It is common that electricity supplier charges according to the generation cost,
i.e., the higher the generation cost the higher the electricity price.
Therefore, it is to the consumers' advantage to reduce electricity consumption at 
high price and hence reduce the electricity bill~\cite{SLL13}.

The smart grid operator and consumers communicate through smart metering devices~\cite{K08,masters2013renewable}. 
A consumer sends in a power request with the power requirement (cf.\ height of request), required duration of service (cf.\ width of request), and the time interval that this request can be served (giving some flexibility).
For example, a consumer may want the dishwasher to operate for one hour during the periods from 8am to 11am. The grid operator upon receiving requests has to schedule them in their respective time intervals using the minimum energy cost. The \emph{load} of the grid at each timeslot is the sum of the power requirements of all requests allocated to that timeslot. The \emph{electricity cost} is modeled by a convex function on the load,
in particular we consider the cost to be the $\alpha$-th power of the load, where $\alpha > 1$ is some constant.
Typically, $\alpha$ is small, e.g., $\alpha=2$~\cite{samadi2010optimal,djurovic2012simplified}.

\runtitle{Previous work.}
Koutsopoulos and Tassiulas~\cite{KT11} has formulated
a similar problem to our problem where
the cost function is piecewise linear.
They show that the problem is NP-hard, and their proof
can be adapted to show the NP-hardness of the general
problem studied in this paper~\cite{BurceaHLWY16}.
Burcea et al.~\cite{BurceaHLWY16} gave polynomial time optimal algorithms for the case of unit height (cf.\ unit power requirement) and unit width (cf.\ unit duration).
Feng et al. \cite{DBLP:conf/cocoa/FengXZ15} have claimed that a simple greedy algorithm is 2-competitive for the unit case and $\alpha = 2$.
However, as to be described below in Lemma~\ref{lm:greedy_lb}, there is indeed a counter example that the greedy algorithm is at least 3-competitive.
This implies that it is still an open question to derive online algorithms for the problem.
Salinas et al.~\cite{SLL13} considered a multi-objective problem
to minimize energy consumption cost and maximize some utility.
A closely related problem is to manage the load
by changing the price of electricity over time~\cite{CK10,MZZ+13,MWJ+10,FangUZ+15a}.
Heuristics have also been developed for demand side management~\cite{LSS12}.
Other aspects of smart grid have also been considered,
e.g., communication~\cite{LQ10,CNX+13,LL13,LJC13}, security~\cite{MYR13,LJC13}.
Reviews of smart grid can be found in~\cite{HG10,IA09,LSM10,IEEE12,FMX+12}.

The main combinatorial problem we defined in this paper
has analogy to the traditional load balancing problem~\cite{A98}
and machine minimization problem~\cite{DBLP:conf/soda/ChenMS16,chuzhoy2004machine,cieliebak2004scheduling,DBLP:conf/fsttcs/Saha13}
but the main differences are the objective being maximum load and 
jobs are unit height~\cite{DBLP:conf/soda/ChenMS16,chuzhoy2004machine,cieliebak2004scheduling,DBLP:conf/fsttcs/Saha13}.
Minimizing maximum load has also been looked at in the context of smart grid~\cite{alamdari2013smart,karbasioun2013power,tang2013smoothing,yaw2014exact,yaw2014peak},
some of which further consider allowing reshaping of the jobs~\cite{alamdari2013smart,karbasioun2013power}.
As to be discussed in Section~\ref{sec:prelim}, our problem is more difficult than minimizing
the maximum load.
Our problem also has resemblance to the dynamic speed scaling problem~\cite{A10,YDS95,BellW15}
and our algorithm has employed some techniques there.

As to be discussed, our problem is closely related to the non-preemptive machine minimization problem~\cite{chuzhoy2004machine,cieliebak2004scheduling},
which has been claimed to be solved optimally in asymptotically sense for the online setting~\cite{DBLP:conf/fsttcs/Saha13}.
We provide an alternative asymptotically optimal competitive algorithm for the non-preemptive machine minimization problem. 
More precisely, we show that our algorithm for the smart grid problem can also solve the non-preemptive machine minimization problem with asymptotically optimal competitive ratio. A more detailed discussion is in Section~\ref{sec:peak}.

\runtitle{Our contribution.}
In this paper, we consider a demand response optimization problem minimizing the total electricity cost
and study its relation with other scheduling problems.
We propose the first online algorithm for the general problem with worst case competitive ratio,
which is polylogarithm in the 
max-min ratio of the duration 
of jobs 
(Theorem~\ref{thm:general} in Section~\ref{sec:general});
and give a lower bound for any online algorithm. 
Interestingly, the ratio depends on the max-min width ratio but not the max-min height ratio.
The algorithm is based on an $O(1)$-competitive online algorithm for jobs with uniform duration (Section~\ref{sec:general_unit_width}). 
We also propose $O(1)$-competitive online algorithms for some special cases
(Section~\ref{sec:special}).
In addition,
we show that 
the problem is fixed parameter tractable by 
proposing the first fixed parameter exact algorithms for the problem;
and derive lower bounds on the running time
(Section~\ref{sec:exact}).
Table~\ref{tab:summary} gives a summary of our results. 
Interestingly, our online algorithm and exact algorithms depend on the variation of the job widths but not the variation of the job heights.

We further show that our online algorithms and exact algorithms can be adapted to the objective of minimizing the peak electricity cost,
as well as the related problem of non-preemptive machine minimization.
Our online algorithms are asymptotically optimal for both problems (Section~\ref{sec:peak_online}), 
with competitive ratio being logarithm in the max-min ratio of the job duration.
In addition, we show that both problems are fixed-parameter tractable (Section~\ref{sec:other_fpt}).

\comment{
Moreover, our online algorithms and exact algorithms with slight modifications can be applied for a smart grid scheduling problem with respect to minimizing peak electricity cost and non-preemptive machine minimization problem.
For both problems, our online algorithms are asymptotically optimal (Section~\ref{sec:peak}), with the competitive ratio being logarithm in the max-min ratio of the duration.
On the other hand, we show that both problems are also fixed-parameter tractable (Section~\ref{sec:other_fpt}).
}


Technically speaking,
our online algorithms are based on identifying a relationship with
the dynamic speed (voltage) scaling ({$\dvs$}) problem~\cite{YDS95}.
The main challenge, even when jobs have uniform width or uniform height,
is that in time intervals where the ``workload'' is low, the optimal {$\dvs$} schedule may have
much lower cost than the optimal {$\grid$} schedule because jobs in {$\dvs$} schedules
can effectively be stretched as flat as possible while jobs in {$\grid$} schedules
have rigid duration and cannot be stretched.
In such case, it is insufficient to
simply compare with the optimal {$\dvs$} schedule.
Therefore, our analysis is divided into two parts:
for high workload intervals, we compare with the optimal {$\dvs$} schedule;
and for low workload intervals, we directly compare with the optimal {$\grid$} schedule
via a lower bound on the total workload over these intervals
(Lemmas~\ref{thm:algV_cost} and~\ref{thm:horizontal_framework}).
For jobs with arbitrary width,
we adopt the natural approach of classification based on job width.
We then align the ``feasible interval'' of each job in a more uniform way
so that we can use the results on uniform width (Lemma~\ref{thm:relax_load}).

In designing exact algorithms we use interval graphs to represent the jobs
and the important notion maximal cliques to partition
the time horizon into disjoint windows.
Such partition usually leads to optimal substructures;
nevertheless, non-preemption makes it trickier and requires a smart way to handle jobs spanning multiple windows.
We describe how to handle such jobs without adding a lot of overhead.
\comment{
We remark that our approach \fhl{for both online and exact algorithms} can be adapted to other objective functions like minimizing maximum load, including the machine minimization problem.
}

\runtitle{Organization of the paper.}
We define the problem and provide some basic observations in Section~\ref{sec:prelim}.
The online algorithms for uniform time duration and arbitrary power requirement are developed in Section~\ref{sec:general_unit_width} and are extended for solving the general case in Section~\ref{sec:general}.
The lower bound of online algorithms is provided in Section~\ref{sec:general_lb}.
Several special cases regarding uniform power requirement are discussed in Section~\ref{sec:special}.
We design fixed-parameter exact algorithms in Section~\ref{sec:exact} and derive a lower bound for the running time in Section~\ref{sec:exact_lb}.
In Section~\ref{sec:peak}, we extend our online and exact algorithms to the objective of maximum load and 
the related non-preemptive machine minimization problem.
We conclude the paper in Section~\ref{sec:conclusion}.
\begin{table}[t]
\begin{center}
\begin{tabular}{|c|c|c|}
\hline
\bf{Width} & \bf{Height} & \bf{Ratio} \\ \hline \hline

\multirow{2}{*}{Unit} & \multirow{2}{*}{Arbitrary} & $2^\alpha \cdot (8(e+e^2)^\alpha+1)$-competitive \\ \cline{3-3}
& & $2^{\alpha+1}$-approximate \\ \hline

Uniform & Arbitrary & $12^\alpha \cdot (8(e+e^2)^\alpha + 1)$-competitive \\ \hline
Arbitrary & Arbitrary & $\Theta(\log^\alpha (\frac{\wmax}{\wmin}))$-competitive \\ \hline
Unit & Uniform & $\min((4\alpha)^\alpha/2 +1, 2^\alpha \cdot (8(e+e^2)^\alpha+1))$-competitive \\ \hline

\multirow{2}{*}{Arbitrary} & \multirow{2}{*}{Uniform} & $((8\alpha)^\alpha/2 + 2^\alpha)$-competitive \\
& & agreeable deadline \\ \hline

\end{tabular}
\caption{Summary of our online results or total electricity cost.}
\label{tab:summary}
\end{center}
\end{table}

%% file: prelim.tex
\section{Definitions and preliminaries}
\label{sec:prelim}


\runtitle{The input.}
The time is labeled from $0$ to $\tau$
and we consider events (release time, deadlines) occurring at integral time.
We call the unit time $[t,t+1)$ \emph{timeslot} $t$.
We denote by $\JS$ a set of input jobs in which each job $J$
comes with \emph{release time} $\rel(J)$,
\emph{deadline} $\dl(J)$,
\emph{width} $\w(J)$ representing the duration required by~$J$, and
\emph{height} $\h(J)$ representing the power required by~$J$.
We assume $\rel(J)$, $\dl(J)$, $\w(J)$, and $\h(J)$ are integers.
The \emph{feasible interval}, denoted by $\interval(J)$,
is defined as the interval $[\rel(J),\dl(J))$ and
we say that~$J$ is \emph{available} during $\interval(J)$.
We denote by $|\interval|$ the length of an interval $\interval$, i.e., $|\interval| = t_2 - t_1$ where $\interval = [t_1, t_2)$.
We define the \emph{density} of $J$, denoted by $\den(J)$, to be $\frac{\h(J) \cdot \w(J)}{|\interval(J)|}$.
Roughly speaking, the density signifies the average load required by the job
over its feasible interval.
We then define the ``average'' load at any time $t$ as
$\avg(t) = \sum_{J:t\in\interval(J)} \den(J)$.
In our analysis, we have to distinguish timeslots with high and low average load.
Therefore, for any $h > 0$, we define $\IL{h}$ and $\IS{h}$ to be set of timeslots where
the average load $\avg(t)$ is larger than $h$ and at most $h$, respectively.
Note that $\IL{h}$ and $\IS{h}$ do not need to be contiguous.

In Section~\ref{sec:general}, we consider an algorithm that classifies jobs 
according to their 
widths.
To ease discussion, we let 
$\wmax$ and $\wmin$ be
the 
maximum and minimum width over all jobs, respectively.
We further define the max-min ratio of width, denoted by $\Kw$,
to be $\Kw=\frac{\wmax}{\wmin}$.
\comment{
We further define the logorithm of the max-min ratio 
of width, denoted by $\logKw$,
to be $\logKw = \lceil \log \frac{\wmax}{\wmin} \rceil$,
which is the number of classes to be defined.
}
Without loss of generality, we assume that 
$\wmin = 1$. 
We say that a job $J$ is in \emph{class} $\class_{p}$ if and only if
$2^{p-1} < \w(J) \leq 2^{p}$ 
for any $0 \leq p \leq \ceilLogKw$.

\runtitle{Feasible schedule.}
A \emph{feasible} schedule $\sch$ has to assign for each job $J$ a \emph{start time} $\stime(\sch, J) \in \mathbb{Z}$
meaning that $J$ runs during $[\stime(\sch,J), \etime(\sch,J))$, 
where the \emph{end time} $\etime(\sch,J) = \stime(\sch,J) + \w(J)$,
and $[\stime(\sch,J), \etime(\sch,J)) \subseteq \interval(J)$.
Note that this means preemption is not allowed.
The \emph{load} of $\sch$ at time $t$, denoted by $\load(\sch,t)$
is the sum of the height (power request) of all jobs running at $t$, 
i.e., $\load(\sch,t) = \sum_{J:t\in [\stime(\sch,J),\etime(\sch,J))} \h(J)$.
We drop $\sch$ and use $\load(t)$ when the context is clear.
For any algorithm~$\Alg$, we use $\Alg(\JS)$ to denote the schedule of $\Alg$ on $\JS$.
We denote by $\Opt$ the optimal algorithm.

The cost of a schedule $\sch$ is the sum of the $\alpha$-th power of the load over all time,
for a constant $\alpha> 1$,
i.e., $\cost(\sch) = \sum_t (\load(\sch,t))^\alpha$.
For a set of timeslots $\I$ (not necessarily contiguous), we denote by $\cost(\sch,\I) = \sum_{t \in \I} (\load(\sch,t))^\alpha$.
Our goal is to find a feasible schedule~$\sch$ such that $\cost(\sch)$ is minimized.
We call this the $\grid$ problem.

\runtitle{Online algorithms.}
In this paper, we consider 
online algorithms, where
the job information is only revealed at the time the job is released; the algorithm has to decide which jobs to run at the current time without future information and
decisions made cannot be changed later.
Let~$\Alg$ be an online 
algorithm.
We say that $\Alg$ is $c$-competitive 
if for all input job sets $\JS$, we have
$\cost(\Alg(\JS)) \leq c \cdot \cost(\Opt(\JS))$.
In particular, we consider non-preemptive algorithms where a job cannot be preempted
to resume/restart later.

\runtitle{Special input instances.}
We consider various special input instances.
A job $J$ is said to be unit-width (resp.\ unit-height) if $\w(J)=1$ (resp.\ $\h(J)=1$).
A job set is said to be uniform-width (resp.\ uniform-height) 
if the width (resp.\ height) of all jobs are the same.
\comment{
jobs with uniform width (Section~\ref{sec:general_unit_width}),
jobs with uniform height and unit width (Section~\ref{sec:special_unit}), 
and jobs with uniform height, arbitrary width, agreeable deadlines (Section~\ref{sec:horizontal_agreeable}).
}
A job set is said to have \emph{agreeable deadlines} if
for any two jobs $J_1$ and $J_2$, $\rel(J_1) \leq \rel(J_2)$ implies $\dl(J_1) \leq \dl(J_2)$.

\comment{
time run $0$ to $\tau$, events happen at integral time, timeslot $t$ refers to $[t,t+1)$

$\JS$, $\class_{p,q}$ width class $p$ and height class $q$, $\relaxS$, $\shrinkS$, $\convert$

$\w(J)$, $\h(J)$, $\rel(J)$, $\dl(J)$, $\den(J)$, $\Kw$, $\Kh$, $\logKw$, $\logKh$

a job available at time $t$

$\stime(\sch, J)$, $\etime(\sch, J)$, 

$\Alg(\JS)$, $\AlgU(\JS)$, $\Opt(\JS)$

$\load(\sch, t)$, e.g., $\load(\Alg(\JS), t)$, $\load(\Opt(\JS^*), t))$, $\sch$ can be dropped for $\load()$ if context is clear

$\cost(\sch, I)$, and $\cost(\sch)$, where $I$ is set of intervals

uniform job set

$\I$, $\IL{h}$, $\IS{h}$

talks about $\AVR$ and state a lemma of its competitive ratio.

EDF

talks about online algorithm, jobs arrive at beginning of timeslot, and then decision is made

agreeable deadlines
}


\runtitle{Relating to the speed scaling problem.}
The {$\grid$} problem resembles the dynamic speed scaling ({$\dvs$}) problem~\cite{YDS95}
and we are going to refer to three algorithms for the {$\dvs$} problem, namely,
the $\YDS$ algorithm which gives an optimal algorithm for the {$\dvs$} problem,
the online algorithms called $\BKP$ and $\AVR$.
We first recap the {$\dvs$} problem and the associated algorithms.
In the {$\dvs$} problem, jobs come with release time $\rel(J)$, deadline $\dl(J)$, and a work requirement $\work(J)$.
A processor can run at speed $s \in [0,\infty)$ and consumes energy in a rate of $s^\alpha$,
for some $\alpha > 1$. 
The objective is to complete all jobs by their deadlines using the minimum total energy.
The main differences of {$\dvs$} problem to the {$\grid$} problem include 
(i) jobs in {$\dvs$} can be preempted while preemption is not allowed in our problem;
(ii) as processor speed in {$\dvs$} can scale, a job can be executed 
for varying time duration as long as the total work is completed 
while in our problem a job must be executed for a fixed duration given as input;
(iii) the work requirement $\work(J)$ of a job $J$ in {$\dvs$} can be seen as
$\w(J) \times \h(J)$ for the corresponding job in {$\grid$}.

With the resemblance of the two problems, we make an observation about
their optimal algorithms.
Let $\Opt_D$ and $\Opt_G$ be the optimal algorithm for the {$\dvs$} and {$\grid$} problem,
respectively.
Given a job set $\JS_G$ for the {$\grid$} problem, 
we can convert it into a job set $\JS_D$ for {$\dvs$} 
by keeping the release time and deadline for each job and 
setting the work requirement of a job in $\JS_D$ to 
the product of the width and height of the corresponding job in $\JS_G$.
Then we have the following observation.

\comment{
\begin{itemize}
\item $\Opt_D$ and $\Opt_G$ for {$\dvs$} and {$\grid$} problem respectively
\item Given $\JS_G$, define $\JS_D$, basically, the size of a job $J_D$ is set to be $\w(J_G)\times\h(J_G)$ while
  the release time and deadline remain the same.
\end{itemize}
}

\begin{observation}
\label{thm:opt_dvs_grid}
Given any schedule $\sch_G$ for $\JS_G$, we can convert $\sch_G$ into a feasible schedule $\sch_D$ for $\JS_D$
such that $\cost(\sch_D(\JS_D)) \leq \cost(\sch_G(\JS_G))$; implying that 
$\cost(\Opt_D(\JS_D)) \leq \cost(\Opt_G(\JS_G))$.
\end{observation}

\begin{proof}
Consider any feasible schedule $\sch_G$.
At timeslot $t$, suppose there are $k$ jobs 
scheduled and their sum of heights is $H$.
The schedule for $\sch_D$ during timeslot~$t$ can be obtained by
running the processor at speed $H$ and the jobs time-share the processor 
in proportion to their height.
This results in a feasible schedule with the same cost
and the observation follows.
\end{proof}

It is known that the online algorithm $\AVR$ for the {$\dvs$} problem is
$\frac{(2\alpha)^\alpha}{2}$-competitive~\cite{YDS95}.
Basically, at any time $t$, $\AVR$ runs the processor at a speed which is the sum of the
densities of jobs that are available at $t$.
By Observation~\ref{thm:opt_dvs_grid}, we have the following corollary.
Note that it is not always possible to convert a feasible schedule for the {$\dvs$} problem
to a feasible schedule for the {$\grid$} problem easily.
Therefore, the corollary does not immediately solve the {$\grid$} problem
but as to be shown it provides a way to analyze algorithms for {$\grid$}.

\begin{corollary}
\label{thm:AVR}
For any input $\JS_G$ and the corresponding input $\JS_D$,
$\cost(\AVR(\JS_D)) \leq \frac{(2\alpha)^\alpha}{2} \cdot \cost(\Opt_G)$.
\end{corollary}

The online algorithm $\BKP$ proposed by Bansal et al.~\cite{DBLP:journals/jacm/BansalKP07} for {$\dvs$} problem is $8e^\alpha$-competitive {with respect to total cost}. 
Let $\load(\BKP,t)$ denote the speed of $\BKP$ at time $t$. 
$\load(\BKP,t) = \max_{t^\prime>t} \frac{\work(t, [et-(e-1)t^\prime,t^\prime))}{t^\prime-t}$ where $\work(t,I)$ denotes the total work of jobs $\Job{}$ with $\interval(J) \subseteq I$ and $\rel(J)\leq t$. 
That is, $\BKP$ chooses the interval $I^\star = [t'',t')$ which has maximal released 
average total work and $(t' - t''):(t' - t) = e:1$ and uses $\frac{\work(t,I^\star)}{|I^\star|/e}$ as the speed at $t$.
By Observation~\ref{thm:opt_dvs_grid} we have the following corollary:
\begin{corollary}
\label{thm:BKP}
For any input $\JS_G$ and the corresponding input $\JS_D$,
$\cost(\BKP(\JS_D)) \leq 8e^\alpha \cdot \cost(\Opt(\JS_D))\leq 8e^\alpha \cdot cost(\Opt(\JS_G))$.
\end{corollary}


\emph{Remark:}
One may
consider the non-preemptive {$\dvs$} problem as the reference of the {$\grid$} problem.
However, given a job set $\JS_G$ and the corresponding $\JS_D$, $\cost(\Opt_D(\JS_D))$ may not necessarily lower than $\cost(\Opt_G(\JS_G))$, where $\Opt_D$ here is the optimal algorithm for non-preemptive {$\dvs$}.
There is an instance shows the optimal cost of {$\grid$} is smaller.
The instance contains two jobs.
One has release time 0, deadline 3, width 3 and height 1.
The other has release time 1, deadline 2, width 1 and height 1.
Both jobs can only schedule at their release time in {$\grid$} since their widths are the same as the lengths of their feasible intervals.
The optimal cost of {$\grid$} is $1^\alpha + 2^\alpha + 1^\alpha = 2^\alpha + 2$.
Whereas the optimal cost of non-preemptive {$\dvs$} is $2^\alpha + 2^\alpha = 2 \cdot 2^\alpha$.
This is because the schedule uses speed 2 and runs the longer job with 1.5 time units and the shorter job with 0.5 time units.
The optimal cost of {$\grid$} is lower when $\alpha > 1$.
Therefore, it is unclear how we may use the results on non-preemptive {$\dvs$} problem and so we would stick with the preemptive {$\dvs$} algorithms.

\runtitle{Relating to minimizing maximum cost.}
The problem of minimizing maximum cost over time (min-max) has been studied before~\cite{yaw2014exact}.
We note that there is a polynomial time reduction 
of the decision version of the min-max problem
to that of the min-sum problem (the {$\grid$} problem we study in this paper)
for a large enough $\alpha$.
In particular, one can show that with $\alpha > (\tau-1)(2 \sum_{J \in \JS} \h(J)+1)$,
the maximum load would dominate the load in other timeslots and we would be
able to solve the min-max problem if we have a solution for the min-sum problem
on $\alpha$.

On the other hand, minimizing the maximum cost does not necessarily 
minimize the total cost.
For example, consider an input of three jobs $J_1$, $J_2$ and $J_3$
where $\interval(J_1) = [0, 2^\alpha)$, $\h(J_1) = 1$, $\w(J_1) = 2^\alpha$;
$\interval(J_2) = [2^\alpha, 2^\alpha+1)$, $\h(J_2) = 3$, $\w(J_2) = 1$; and
$\interval(J_3) = [0, 2^{\alpha+1})$, $\h(J_3) = 1$, $\w(J_3) = 2^\alpha$.
Note that only $J_3$ has flexibility where it can be scheduled.
To minimize the maximum cost over time, we would schedule $J_3$ to start at time $0$
and achieve a maximum load of $3$.
This gives a total cost of $2^\alpha \cdot 2^\alpha + 3^\alpha = 4^\alpha + 3^\alpha$.
However, to minimize the total cost, we would schedule $J_3$ to start at time $2^\alpha$
giving a total cost of $2^\alpha + 4^\alpha + (2^\alpha - 1) = 4^\alpha + 2^{\alpha+1} - 1$,
which is smaller than $4^\alpha + 3^\alpha$ when $\alpha > 1$.

\runtitle{Lower bound on Greedy.}
In~\cite{DBLP:conf/cocoa/FengXZ15}, the greedy algorithm that assigns a job to a timeslot with the minimum load is considered.
It is claimed in the paper that the greedy algorithm is $2$-competitive
on the online-list model
and for the case where the load of a timeslot $t$ is $\load(t)^2$,
jobs are of unit length and height and the feasible timeslots of a job is a set of (non-contiguous) timeslots that the job can be assigned to.
We show a counter-example to this claim and show that Greedy is at least $3$-competitive.
This implies that it is still an open question to derive online algorithms for the {$\grid$} problem.

\begin{lemma}
\label{lm:greedy_lb}
Greedy is no better than $3$-competitive for the online-list model when $\alpha = 2$.
\end{lemma}

\begin{proof}
Let $k$ be an arbitrarily large integer.
The adversary works in $k$ rounds and all the jobs released are of width and height $1$.
In the $i$-th round, where $1 \leq i < k$, the adversary releases $2^{k-i}$ jobs;
and in the $k$-th round (the final one), the adversary releases two jobs.
In the first round, the feasible timeslots of each job released are $[1,2^k]$.
In the $i$-th round, where $2 \leq i \leq k$, the feasible timeslots of each job released are all the timeslots that 
Greedy has assigned jobs in the $(i-1)$-th round.
We claim that the total cost of Greedy is $3\cdot 2^k -4$
and the total cost of the optimal algorithm is $2^k$.
Therefore, the competitive ratio of Greedy is arbitrarily close to $3$ with an arbitrarily large integer $k$.

We first analyze Greedy.
Since Greedy always assigns to a timeslot with the minimum load,
in the first round, Greedy assigns jobs to $2^{k-1}$ timeslots with each job to a different timeslot.
These $2^{k-1}$ timeslots will be the feasible timeslots for the $2^{k-2}$ jobs in the second round.
Using a similar argument, we can see that in each round, the number of feasible timeslots is twice the number of jobs released in that round.
In addition, before the $i$-th round,
the load of each feasible timeslot is $i-1$
and Greedy adds a load of 1 to each timeslot that it assigns a job,
making the load become $i$.
Therefore, the total cost of Greedy is
$\sum_{i=1}^{k-2} (i^2 \cdot 2^{k-i-1}) + k^2 \cdot 2 = 3 \cdot 2^k - 4$.
On the other hand, we can assign jobs released in a round to the timeslots that are not feasible timeslots for later rounds
since in the $i$-th round, the number of feasible timeslots is $2^{k-i+1}$ and the number of jobs released is $2^{k-i}$.
Therefore, in the optimal schedule, the load of each timeslot is exactly 1 and the total cost is $2^k$.
\end{proof}

%% file: verticalUniform.tex
\section{Online algorithm for uniform width jobs}
\label{sec:general_unit_width}
To handle jobs of arbitrary width and height,
we first study the case when jobs have uniform width 
(all jobs have the same width $\w \geq 1$).
The proposed algorithm~$\AlgUV$ (Section~\ref{sec:uniform_width}) is based on 
a further restricted case of unit width, i.e., $\w = 1$ (Section~\ref{sec:vertical}).

\comment{
In this section, we propose an online scheduling algorithm for jobs with arbitrary height and uniform width. 
For the jobs with uniform width, we separate the jobs into to groups, ``tight jobs'' and ``loose jobs'', by their length of feasible interval. We schedule these two groups of jobs separately. 
For tight jobs, we show that a constant competitive ratio can be achieved by a simple strategy. 
For loose jobs, we transform the input such that they can be treated like a set of jobs with unit width.
We first show an online algorithm for jobs with unit width. 
}

\input{vertical}

\subsection{Uniform width and arbitrary height}
\label{sec:uniform_width}

In this section, we consider jobs with uniform width $\w$ and arbitrary height.
The idea of handling uniform width jobs is to treat them as if they were unit width,
however, this would mean that jobs may have release times or deadlines at non-integral time.
To remedy this, we define a procedure $\convertfi$ to align the feasible intervals
(precisely, release times and deadlines)
to the new time unit of duration~$\w$.

Let $\uwJS$ be a uniform width job set. 
We first define the notion of ``tight'' and ``loose'' jobs.
A job~$J$ is said to be \emph{tight} if $|\interval(J)| < 2\w$; otherwise, it is \emph{loose}. 
Let $\uwJSt$ and $\uwJSl$ be the disjoint subsets of tight and loose jobs of $\uwJS$, respectively.
We design different strategies for tight and loose jobs.
As to be shown, tight jobs can be handled easily 
by starting them at their release times.
For any loose job, we modify it via Procedure $\convertfi$
such that its release time and deadline is a multiple of $\w$.
With this alternation, we can treat the jobs as unit width and
make scheduling decisions at time multiple of $\w$.

\comment{
For set $\uwJS$ of jobs with uniform width $\w$, we partition $\uwJS$ into two groups, tight jobs and loose jobs. A job~$J$ is said to be \emph{tight} if $|\interval(J)| < 2\w$; otherwise, it is \emph{loose}.
Section~\ref{sec:vertical}
Let $\uwJSt$ and $\uwJSl$ be the disjoint subsets of tight and loose jobs of $\uwJS$, respectively. 
For tight jobs and loose jobs we have different strategy. For tight jobs, they are executed once they are released. We prove this strategy is constant competitive later. 
For loose jobs, we modify the release time and deadline by Procedure $\convertfi$ such that each job has new release time and deadline at $i\cdot \w$ for some integer $i$. 
Then we run $\AlgV$ on the set of loose jobs with modified release time and deadline, and make decision only at timeslots $i\cdot \w$ where $i$ is integer.
}

\runtitle{Procedure $\convertfi$.} 
Given a loose job set $\uwJSl$ in which $\w(J)=w$ and $|\interval(J)| \geq 2\cdot w$ 
$\forall J\in\uwJSl$. 
We define the procedure $\convertfi$ to transform each loose job $J\in \uwJSl$ 
into a job $\uwaJ$ with release time and deadline ``aligned'' as follows. We denote the resulting job set by $\uwaJS$. 
\begin{itemize}
\item $\rel(\uwaJ) \leftarrow \min_{i\geq 0} \{i\cdot w\mid i\cdot w\geq\rel(J)\}$;
\item $\dl(\uwaJ) \leftarrow \max_{i\geq 0} \{i\cdot w\mid i\cdot w\leq\dl(J)\}$.
\end{itemize}

\begin{observation}
\label{obs:uwafi}
For any job $J\in \uwJSl$ and the corresponding $\uwaJ$,
\text{\rm (i)} $\frac{1}{3}\cdot|\interval(J)|<|\interval(\uwaJ)| \leq |\interval(J)|$;
\text{\rm (ii)} $|\interval(\uwaJ)|\geq w$;
\text{\rm (iii)} $\interval(\uwaJ)\subseteq\interval(J)$.
\end{observation}

Notice that after $\convertfi$, the release time and deadline of each loose job are aligned to timeslot $i_1\cdot w$ and $i_2 \cdot w$ for some integers $i_1<i_2$. 
By Observation~\ref{obs:uwafi}, a feasible schedule of $J^\prime$ is also a feasible schedule of~$J$. 
Furthermore, after $\convertfi$ all jobs are released at time which is a multiple of $\w$.
Hence,
the job set $\uwaJS$ can be treated as job set with unit width, 
where each unit has duration $\w$ instead of $1$.  


As a consequence of altering the feasible intervals,
we introduce two additional procedures that convert associated schedules.
Given a schedule $\uwSch$ for job set $\uwJSl$, $\alignS$ converts it to a schedule $\uwaSch$ for the corresponding job set $\uwaJS$. The other procedure $\freeS$ takes a schedule $\uwaSch$ for a job set $\uwaJS$ and converts it to a schedule $\uwSch$ for $\uwJSl$.

\runtitle{Transformation $\alignS$.}
$\alignS$ transforms $\uwSch$ into $\uwaSch$ by shifting the execution interval of every job $J\in\uwJSl$.
\begin{itemize}
\item $\stime(\uwaSch,\uwaJ) \leftarrow \min\{ \dl(\uwaJ)-\w(\uwaJ), \min_{i\geq 0} \{i\cdot w\mid i\cdot w\geq\stime(\sch,J)\}\}$;
\item $\etime(\uwaSch,\uwaJ) \leftarrow \stime(\uwaSch,\uwaJ) + \w(\uwaJ)$.
\end{itemize}

\begin{observation}
\label{obs:aling_overlap}
\label{obs:align_feasible}
\label{thm:align_load}
Consider any schedule $\uwSch$ for $\uwJSl$ and the schedule $\uwaSch$ for $\uwaJS$ constructed by $\alignS$. The following properties hold:
\text{\rm (i)} For any job $\uwJ\in\uwJSl$ 
and the corresponding $\uwaJ$, $\stime(\uwaJ) > \stime(\uwJ)-w$ and $\etime(\uwaJ) < \etime(\uwJ)+w$; 
\text{\rm (ii)} $\uwaSch$ is a feasible schedule for $\uwaJS$; and 
\text{\rm (iii)} At any time~$t$, $\load(\uwaSch, t) \leq  \load(\uwSch, t) + \load(\uwSch, t-(w-1)) + \load(\uwSch, t+(w-1))$.
\end{observation}

\begin{proof}
(ii) By $\alignS$, $\stime(\uwaSch,\uwaJ)\leq \dl(\uwaJ)-\w(\uwaJ)$. 
Also, $|[\stime(\uwaSch,\uwaJ),\etime(\uwaSch,\uwaJ))| = \w(\uwaJ)$. 
Hence $[\stime(\uwaSch,\uwaJ),\etime(\uwaSch,\uwaJ))\subseteq \interval(\uwaJ)$. 
That is, $\uwaSch$ is a feasible schedule for both $\uwaJS$ and $\uwJ$.

(iii) By (i), $\stime(\uwaJ) > \stime(\uwJ)-w$ and $\etime(\uwaJ) < \etime(\uwJ)+w$ for each $\uwJ$. 
Hence, for any timeslot $t$, for each job $\uwJ$ with $[\stime(\uwSch,\uwJ),\etime(\uwSch,\uwJ))\cap[t-(w-1), t+(w-1))=\emptyset$, $t\notin[\stime(\uwaSch,\uwaJ),\etime(\uwaSch,\uwaJ))$. 
On the other hand, consider the jobs $J$ that $[\stime(\uwJ),\etime(\uwJ))\cap[t-(w-1), t+(w-1))\neq\emptyset$. 
Since $|[\stime(\uwJ),\etime(\uwJ))|=w$, at least one of the timeslots $t-(w-1)$, $t$, or $t+(w-1)$ is in $[\stime(\uwJ),\etime(\uwJ))$. 
Hence we can capture $\load(\uwaSch, t)$ by $\load(\uwSch, t) + \load(\uwSch, t-(w-1)) + \load(\uwSch, t+(w-1))$.
\end{proof}


\begin{corollary}
\label{thm:align_cost}
Using $\alignS$ to generate $\uwaSch$ given $\uwSch$, we have $\cost(\uwaSch) \leq 3^\alpha \cdot\cost(\uwSch)$.
\end{corollary}
\begin{proof}
By Observation~\ref{thm:align_load} (iii), $\cost(\uwaSch) = \sum_t \load(\uwaSch,t)^\alpha\leq \sum_t (3 \cdot\load(\uwSch,t))^\alpha = 3^\alpha\cdot\cost(\uwSch)$.
\end{proof}

\begin{lemma}
\label{thm:align_opt}
$\cost(\Opt(\uwaJS))\leq 3^\alpha\cdot\cost(\Opt(\uwJSl))$.
\end{lemma}

\begin{proof}
Consider set of loose jobs $\uwJSl$ with uniform width and the corresponding $\uwaJS$. Given $\Opt(\uwJSl)$, there exists schedule $S(\uwaJS)$ generated by $\alignS$. 
By Lemma~\ref{thm:align_cost}, $\cost(S(\uwaJS))\leq 3^\alpha\cdot\cost(\Opt(\uwJSl))$. Hence, $\cost(\Opt(\uwaJS))\leq\cost(S(\uwaJS))\leq 3^\alpha\cdot\cost(\Opt(\uwJSl))$.
\end{proof}

\runtitle{Transformation $\freeS$.}
$\freeS$ transforms $\uwaSch$ into $\uwSch$.
\begin{itemize}
\item $\stime(\uwSch,J) \leftarrow \stime(\uwaSch,\uwaJ)$;
\item $\etime(\uwSch,J) \leftarrow \etime(\uwaSch,\uwaJ)$.
\end{itemize}

The feasibility of $\uwaSch$ can be easily proved by Observation~\ref{obs:uwafi}.

\begin{lemma}
\label{thm:free_cost}
Using $\freeS$, we have $\cost(\uwSch)= \cost(\uwaSch)$.
\end{lemma}

\begin{proof}
Since the execution intervals of $J$ and $\uwaJ$ are the same, $\load(\uwSch,t) = \load(\uwaSch,t)$ for all $t$. Hence $\cost(\uwSch)= \cost(\uwaSch)$.
\end{proof}

\runtitle{Online algorithm $\AlgUV$.}
The algorithm takes a job set $\uwJS$ with uniform width $w$ as input and schedules the jobs in $\uwJS$ as follows. Let $\uwJSt$ be the set of tight jobs in $\uwJS$ and $\uwJSl$ be the set of loose jobs in $\uwJS$.
\begin{enumerate}
\item For any tight job $J \in \uwJSt$, schedule $J$ to start at $\rel(J)$.
\item Loose jobs in $\uwJSl$ are converted to $\uwaJS$ by $\convertfi$. 
For $\uwaJS$, we run Algorithm $\AlgV$, 
which is defined in Section~\ref{sec:vertical},
with $\BKP$ as the reference {$\dvs$} algorithm. 
Jobs are chosen in an earliest deadline first (EDF) manner.
\end{enumerate}
Note that the decisions of $\AlgUV$ can be made online.

\runtitle{Analysis of Algorithm $\AlgUV$.}
We analyze the tight jobs and loose jobs separately. 
We first give an observation.

\begin{observation}
\label{obs:subopt}
For any two job sets $\JS_x \subseteq \JS_y$, $\cost(\Opt(\JS_x)) \leq \cost(\Opt(\JS_y))$.
\end{observation}

\begin{proof}
Assume on the contrary that $\cost(\Opt(\JS_y)) < \cost(\Opt(\JS_x))$, we can generate a schedule $S(\JS_x)$ by removing jobs from $\Opt(\JS_y)$ which are not in $\JS_x$. 
It follows that $\cost(S(\JS_x)) \leq \cost(\Opt(\JS_y)) < \cost(\Opt(\JS_x))$, contradicting to the fact that $\Opt(\JS_x)$ is optimal for $\JS_x$.
\end{proof}

\comment{
To analyze the performance of $\AlgUV$,
recall that $\load(\BKP,t) = \max_{t^\prime>t} \frac{\work(t, [et-(e-1)t^\prime,t^\prime))}{t^\prime-t}$ where $\work(t,I)$ denotes the total work of jobs $J$ with $\interval(J) \subseteq I$ and $\rel(J)\leq t$. 
That is, $\BKP$ chooses the interval $I^\star = [t_1,t_2)$ which has maximal released average load and $|[t_1,t_2)|:|[t,t2)|= e:1$ and uses $e\cdot\frac{\work(t,I^\star)}{|I^\star|}$ as the speed at $t$.
}
In the following analysis we say that interval $I = [t_1,t_2)$ is a \emph{$\BKP$ interval of $t$} if $t\in I$ and $(t_2-t_1):(t_2-t)= e:1$. 
The next lemma proves the competitive ratio separately for $\uwJSt$ and $\uwJSl$.

\begin{lemma}
\label{thm:cost_UV1}
\label{thm:cost_UV2}
\text{\rm (i)} $\cost(\AlgUV(\uwJSt))\leq 3^\alpha\cdot\cost(\Opt(\uwJS))$; 
\text{\rm (ii)} $\cost(\AlgUV(\uwJSl))\leq 6^\alpha\cdot(8 (e+e^2)^\alpha+1)\cdot\cost(\Opt(\uwJS))$.
\end{lemma}

\begin{proof}
(i) We prove that any feasible schedule $\Sch$ for tight jobs is $3^\alpha$-competitive. 
We first extend jobs $\Job{} \in \JStight$ to $\alignJob{}$ as follows: $\alignRtime{} = \Rtime{}$, $\alignDline{} = \Dline{}$, $\alignWidth{} = \Dline{}-\Rtime{}$, and $\alignHeight{} = \Height{}$. That is, every job has its width as the length of its feasible interval. 
We denote the resulting job set by $\alignJobsetij{}$. Since each job in $\alignJobsetij{}$ are not shiftable, there is only one feasible schedule for $\alignJobsetij{}$ and it is optimal. 
Thus, $\cost(\Sch(\JStight)) \leq \cost(\Opt(\alignJobsetij{}))$ for any feasible schedule $\Sch$ for $\JStight$.

For each job in $\JStight$, the length of its feasible interval is at most $2\w-1$.
Hence, we can bound the load at any time $t$ of $\Opt(\alignJobsetij{})$ by the loads of constant number of timeslots in $\Sch(\JStight)$. 
Assume that at timeslot $t$ an extended job $\alignJob{}$ is executed. That is, $t \in [\alignRtime{}, \alignDline{})$ since $\alignJob{}$ is not shiftable. 
Consider the job $\Job{}$ corresponding to $\alignJob{}$, the execution interval of $\Job{}$ in any feasible schedule must contains either timeslot $t-(\w-1)$, $t$, or $t+(\w-1)$. 
Hence we can upper bound the load at any time $t$ in $\Opt(\alignJobsetij{})$: $\load(\Opt(\alignJobsetij{}), t) \leq \load(\Opt(\JStight), t-(\w-1)) + \load(\Opt(\JStight), t) + \load(\Opt(\JStight), t+(\w-1))$. 
Therefore, $\cost(\Sch(\JStight)) \leq \cost(\Opt(\alignJobsetij{})) \leq 3^\alpha\cdot \cost(\Opt(\JStight))$.

(ii) 
For $\uwJSl$, we apply $\convertfi$ and get $\uwaJSl$. 
We then run $\AlgV$ and get $\AlgV(\uwaJSl)$, 
which can be viewed as a schedule for unit width jobs. 
We get $\uwSch(\uwJSl)=\AlgV(\uwaJSl)$ by $\freeS$. 
Hence, $\cost(\AlgUV(\uwJSl))=\sum_t \load(\AlgUV(\uwJSl),t)^\alpha=\sum_t \load(\uwSch(\uwJSl),t)^\alpha = \sum_t \load(\AlgV(\uwaJSl),t)^\alpha = \cost(\AlgV(\uwaJSl))$. 
According to Corollary~\ref{thm:V_competitive}, $\cost(\AlgV(\uwaJSl))\leq 2^\alpha\cdot(8\cdot (e+e^2)^\alpha+1)\cdot\cost(\Opt(\uwaJSl))$ by choosing $\BKP$ as reference algorithm. Since $\uwJSl$ is set of loose jobs with uniform width, $\cost(\Opt(\uwaJSl))\leq 3^\alpha\cdot\cost(\Opt(\uwJSl)) \leq 3^\alpha\cdot\cost(\Opt(\uwJS))$ by Lemma~\ref{thm:align_opt} and Observation~\ref{obs:subopt}. Hence, $\cost(\AlgUV(\uwJSl))\leq 2^\alpha\cdot(8\cdot (e+e^2)^\alpha+1)\cdot 3^\alpha\cdot\cost(\Opt(\uwJS))$.
\end{proof}


\begin{theorem}
\label{thm:cost_UV}
$\cost(\AlgUV(\uwJS))\leq 12^\alpha\cdot(8 (e+e^2)^\alpha+1)\cdot\cost(\Opt(\uwJS))$.
\end{theorem}

\begin{proof}
By definition, $\cost(\AlgUV(\uwJS)) = \sum_t \load(\AlgUV(\uwJS),t)^\alpha = \sum_t (\load(\AlgUV(\uwJSt),t)+\load(\AlgUV(\uwJSl),t))^\alpha \leq 2^{\alpha-1} \cdot \sum_t (\load(\AlgUV(\uwJSt),t)^\alpha+\load(\AlgUV(\uwJSl),t)^\alpha) = 2^{\alpha-1}\cdot(\cost(\AlgUV(\uwJSt))+\cost(\AlgUV(\uwJSl)))$. 
By Lemma~\ref{thm:cost_UV1}, $\cost(\AlgUV(\uwJS))\leq 2^{\alpha-1}\cdot (3^\alpha + 6^\alpha\cdot(8(e+e^2)^\alpha+1))\cdot\cost(\Opt(\uwJS))\leq 2^\alpha\cdot 6^\alpha\cdot(8(e+e^2)^\alpha+1)\cdot \cost(\Opt(\uwJS))$.
\end{proof}

%% file: vertical.tex
\subsection{Unit width and arbitrary height}
\label{sec:vertical}


In this section, we consider jobs with unit width and arbitrary height.
We present an online algorithm $\AlgV$ which makes reference to 
an arbitrary feasible online algorithm for the $\dvs$ problem, denoted by $\R$.
In particular, we require that the speed of $\R$ remains the same 
during any integral timeslot, i.e., in $[t,t+1)$ for all integers $t$.
Note that when jobs have integral release times and deadlines,
many known $\dvs$ algorithms satisfy this criteria, including $\YDS$, $\BKP$, and $\AVR$.


Recall in Section~\ref{sec:prelim} how a job set for the $\grid$ problem
is converted to a job set for the $\dvs$ problem.
We simulate a copy of~$\R$ on the converted job set
and denote the speed used by~$\R$ at~$t$ as $\load(\R,t)$.
Our algorithm makes reference to $\load(\R,t)$ but not the jobs run by~$\R$ at~$t$.

\comment{
\runtitle{Transforming to $\dvs$ problem.} We can transform the input of $\grid$ problem to a $\dvs$ problem. 
For job $J$ with width $\w(J)$, height $\h(J)$, release time $\rel(J)$ and deadline $\dl(J)$, it is converted to a job 
$J^\prime$ with work requirement $\work(J^\prime) = \w(J)\times\h(J)$ with release time $\rel(J^\prime)=\rel(J)$ and deadline $\dl(J^\prime)=\dl(J)$.
We say we "run $\sch$ on $\JS$" where $\sch$ is scheduling algorithm and $\JS$ is the input set.
}

\runtitle{Algorithm $\AlgV$.}
For each timeslot~$t$, we schedule jobs to start at $t$ such that $\load(\AlgV,t)$ 
is at least $\load(\R,t)$ or until all available jobs  have been scheduled. 
Jobs are chosen in an EDF manner.

\runtitle{Analysis.}
We note that since $\AlgV$ makes decision at integral time and
jobs have unit width, 
each job is completed before any further scheduling decision is made.
In other words, $\AlgV$ is non-preemptive.
To analyze the performance of $\AlgV$, we first note that $\AlgV$ gives
a feasible schedule (Lemma~\ref{thm:AlgV_feasible}),
and then analyze its competitive ratio (Theorem~\ref{thm:algV_approximate}).

\begin{lemma}
\label{thm:AlgV_feasible}
$\AlgV$ gives a feasible schedule.
\end{lemma}

\begin{proof}
Let $\load(\sch,\I)$ denote the total work done by schedule $\sch$ in $\I$. That is, $\load(\sch,\I) = \sum_{t\in\I}\load(\sch,\I)$. According to the algorithm, for all $\I_t = [0,t)$, $\load(\AlgV,\I_t) \geq \load(\R,\I_t)$. 

Suppose on the contrary that $\AlgV$ has a job $J_m$ missing deadline at $t$. That is, $\dl(J_m)= t$ but $J_m$ is not assigned before $t$. 
By the algorithm, for all $t^\prime\in[0,t)$, $\load(\AlgV,t^\prime) \geq \load(\R,t^\prime)$ unless there are less than $\load(\R,t^\prime)$ available jobs at $t^\prime$ for $\AlgV$. 
Let $t_0$ be the last timeslot in $[0,t)$ such that $\load(\AlgV,t_0)<\load(\R,t_0)$, $\rel(J_m)>t_0$ since all jobs released at or before $t_0$ have been assigned. 
For all $t^\prime\in(t_0,t)$, $\load(\AlgV,t^\prime) \geq \load(\R,t^\prime)$. Also, all jobs $J$ with $\rel(J)\leq t_0$ are finished by $t_0+1$ and jobs executed in $(t_0,t)$ are those released after $t_0$. 
Consider set $\JS_t$ of jobs with feasible interval completely inside $\I=(t_0,t)$ (note that $J_m\in\JS_t$), $\load(S,\I)\geq\sum_{J\in\JS_t \h(J)}$ for any feasible schedule $S$.
Since $\AlgV$ assigns jobs in EDF manner and is not feasible, $\load(\AlgV,\I) < \sum_{J\in\JS_t} \h(J)$. It follows that $\sum_{J\in\JS_t} \h(J) > \load(\AlgV,\I) \geq \load(\R,\I)$. It contradicts to the fact that $\R$ is feasible. Hence, $\AlgV$ finishes all jobs before their deadlines.
\end{proof}



Let $\hmax(\AlgV,t)$ be the maximum height of jobs scheduled at $t$ by $\AlgV$;
we set $\hmax(\AlgV, t)= 0$ if $\AlgV$ assigns no job at $t$.
We first classify each timeslot $t$ into two types: 
(i) $\hmax(\AlgV,t) < \load(\R,t)$, and (ii) $\hmax(\AlgV,t) \geq \load(\R,t)$. 
We denote by $\I_1$ and $\I_2$ the union of all timeslots of Type (i) and (ii), respectively. 
Notice that $\I_1$ and $\I_2$ can be empty 
and the union of $\I_1$ and $\I_2$ covers the entire time line.
The following lemma bounds the cost of $\AlgV$ in each type of timeslots.
Recall that $\cost(\sch,\I)$ denotes the cost of the schedule $\sch$ over the interval $\I$ and $\cost(\sch)$ denotes the cost of the entire schedule.

\begin{lemma}
\label{thm:algV_cost}
The cost of $\AlgV$ satisfies the following properties.
\text{\rm (i)} $\cost(\AlgV,\I_1) \leq 2^\alpha \cdot \cost(\R)$; and 
\text{\rm (ii)} $\cost(\AlgV,\I_2) \leq 2^\alpha \cdot \cost(\Opt)$.
\end{lemma}

\begin{proof}
%

(i) By the algorithm, $\load(\AlgV,t) < \load(\R,t) + \hmax(\AlgV,t) \leq 2\cdot\load(\R,t)$ for $t\in \I_1$. 
It follows that $\cost(\AlgV,\I_1) \leq 2^\alpha\cdot \sum_{t\in\I_1} \load(\R,t)^\alpha = 2^\alpha\cdot \cost(\R,\I_1)\leq 2^\alpha\cdot \cost(\R)$.

(ii) 
By convexity, $\cost(\Opt) \geq \sum_J \h(J)^\alpha$. We can see that $\cost(\Opt) \geq \sum_{t\in\I_2}\hmax(\AlgV,t)^\alpha$. 
According to the algorithm, $\load(\AlgV,t) < \load(\R,t) + \hmax(\AlgV,t) \leq 2\cdot\hmax(\AlgV,t)$ for $t\in \I_2$. 
Hence, $\cost(\AlgV, \I_2) = \sum_{t\in\I_2} \load(\AlgV,t)^\alpha \leq 2^\alpha\cdot\sum_{t\in\I_2}\hmax(\AlgV,t)^\alpha\leq 2^\alpha\cdot\cost(\Opt)$. 
\end{proof}

Notice that $\cost(\AlgV) = \cost(\AlgV,\I_1) + \cost(\AlgV,\I_2)$ since $\I_1$ and $\I_2$ have no overlap. Together with Lemma~\ref{thm:algV_cost} and Observation~\ref{thm:opt_dvs_grid}, we obtain the competitive ratio of $\AlgV$
in the following theorem.

\begin{theorem}
\label{thm:algV_approximate}
Algorithm $\AlgV$ is $2^\alpha\cdot(R+1)$-competitive, where $R$ is the competitive ratio of the reference $\dvs$ algorithm $\R$.
\end{theorem}

There are a number of $\dvs$ algorithms that can be used as the reference algorithm. The only requirement is that
the speed of the reference algorithm within any integral interval $[t, t+1)$ for some integer $t$ should be at most the load of the resulting online algorithm at the corresponding timeslot $t$.
Otherwise, the feasibility of $\AlgV$ cannot be guaranteed. Also, since in our online algorithm we make decision at each integral time $t$, it means if the load of the reference algorithm at $i +\Delta$ is larger than $\load(\R, i)$ for some $0<\Delta <1$, our online algorithm might not be feasible.

The speed of the $\AVR$ and $\YDS$ algorithm only change at release times or deadlines of the jobs so it is valid to use $\AVR$ or $\YDS$ as a reference. 
Note that if we use $\YDS$ as the reference, the algorithm~$\AlgV$ is an offline algorithm since $\YDS$ is an offline algorithm. 
Unlike $\AVR$ and $\YDS$, the speed of $\BKP$ within a timeslot might increase. 
Hence, we need to modify the $\BKP$ algorithm such that it can be used as the reference algorithm. 
In Lemma~\ref{lm:bkp_load}, we show that the speed of $\BKP$ in $[t,t+1)$ is bounded by a constant factor times the speed at~$t$ for any time~$t$.

\begin{lemma}
\label{lm:bkp_load}
For any integral time $t$ and a constant $0< \Delta < 1$, $\load(\BKP,t+\Delta)\leq (1+e)\cdot \load(\BKP,t)$ if the release times and deadlines of jobs are integral.
\end{lemma}

\begin{proof}
Recall that the speed of $\BKP$ at time $t$, $\load(\BKP,t) = \max_{\interval} e\cdot \frac{\work(t,\interval)}{|\interval|}$ where $\interval = [t_1,t_2)$ and $(t_2 - t_1):(t_2 - t) = e:1$. The proof idea is, consider the interval $\interval$ chosen by $\BKP$ corresponding to $t + \Delta$, we can transform it into another interval $\interval^\prime$ which is one of the interval candidate for $t$. We show that $e\cdot\frac{\work(t,\interval^\prime)}{|\interval^\prime|}$ is at least $\frac{1}{1+e}$ times of the speed of $\BKP$ at $t+\Delta$.

Assume that at time $t+\Delta$, $\load(\BKP,t+\Delta) = e\cdot \frac{\work(t,\interval)}{|\interval|}$ where $\interval = [t_1,t_2)$ is chosen by $\BKP$. We can construct $\interval^\prime = [t^\prime_1, t_2)$ such that $(t_2 - t_1^\prime):(t_2 - t) = e:1$ by setting $t_1^\prime = t_2-e(t_2-t)$. It is clear that $\interval \subset \interval^\prime$ since the two intervals have the same right endpoint and $\interval^\prime$ is longer than $\interval$. 
In fact, $|\interval^\prime| = e(t_2-t) = e(t_2-(t+\Delta)) + e\Delta = |\interval|+e\Delta \leq |\interval|+e$. 
Moreover, for any interval candidate, the length must be at least $1$ if the release times and deadlines of the jobs are integral. Otherwise, the interval contains no jobs and the speed is $0$. 
Hence, $|\interval^\prime| \leq (1+e)|\interval|$. By $\BKP$, $\load(\BKP,t) \geq e\cdot \frac{\work(t,\interval^\prime)}{|\interval^\prime|} = e\cdot \frac{\work(t+\Delta,\interval^\prime)}{|\interval^\prime|}$. 
The later equality holds since there is no job released between $t$ and $t+\Delta$. 
Since $\interval \subset \interval^\prime$ and $|\interval^\prime| \leq (1+e)|\interval|$, $e\cdot \frac{\work(t+\Delta,\interval^\prime)}{|\interval^\prime|} \geq e\cdot \frac{\work(t+\Delta,\interval)}{|\interval^\prime|} \geq e\cdot \frac{\work(t+\Delta,\interval)}{(1+e)|\interval|}$. 
Hence, $\load(\BKP,t) \geq \frac{1}{1+e} \cdot \load(\BKP,t+\Delta)$. 
\end{proof}

Lemma~\ref{lm:bkp_load} implies that, although the speeds of $\BKP$ change within $[t, t + 1)$, the speeds are bounded by $(1 + e)$ times of the speed at~$t$. Hence, we can modify $\BKP$ into $\newBKP$ as follows: at integral time~$t$, the speed of $\newBKP$, $\load(\newBKP, t) = (1 + e)\load(\BKP,t)$; at time $t^\prime = t+\Delta$ where~$t$ is integral and $0 < \Delta < 1$, $\load(\newBKP,t^\prime) = \load(\newBKP,t)$. By the modification, the speed of $\newBKP$ remains the same during any integral timeslot, and $\cost(\newBKP) \leq (1 + e)^\alpha \cdot \cost(\BKP)$. As mentioned in Section~\ref{sec:prelim}, the $\BKP$ algorithm is $8\cdot e^\alpha$-competitive.
On the other hand, $\AlgV$ can take an offline $\dvs$ algorithm, 
e.g., the optimal $\YDS$ algorithm, as reference
and returns an offline schedule.
Therefore, we have the following corollary. 

\begin{corollary}
\label{thm:V_competitive}
$\AlgV$ is $2^\alpha\cdot(8\cdot (e+e^2)^\alpha+1)$-competitive, $2^\alpha\cdot(\frac{(2\alpha)^\alpha}{2}+1)$-competitive, and
$2^\alpha\cdot 2$-approximate when the algorithm $\newBKP$, $\AVR$, and $\YDS$ are referenced, respectively.
\end{corollary}

\comment{

Consider online algorithm $\BKP$ proposed by Bansal et al.~\cite{DBLP:journals/jacm/BansalKP07} for $\dvs$ problem. 
If we run $\BKP$ on the vertical jobs set of the $\grid$ problem, load of $\BKP$ does not change between $[t,t+1)$ since all jobs with release time and deadline on integral time. 
By choosing $\BKP$ as $\R$, we get the following corollary. 
We use $\AlgV_\Alg$ to denote the online algorithm choosing $\Alg$ as reference algorithm $\R$.

\begin{corollary}
\label{thm:UV_competitive}
$\AlgV_\BKP$ is $2^\alpha\cdot(8\cdot e^\alpha+1)$-competitive when $\alpha \geq 2$.
\end{corollary}

$\AlgV$ can also run in an offline fashion.
By choosing $\YDS$ as $\R$, we obtain the following corollary.

\begin{corollary}
$\AlgV_\YDS$ is $2^\alpha\cdot 2$-approximate.
\end{corollary}

}

%% file: general.tex
\renewcommand{\uwJS}{{\JS^*}}
\renewcommand{\uwSch}{{\sch^*}}
\renewcommand{\uwJ}{{J^*}}
\renewcommand{\Alg}{{\mathcal{G}}}

\section{Online algorithm for general case}
\label{sec:general}

In this section, we present an algorithm~$\Alg$ for jobs with arbitrary 
width and height.
We first transform job set $\JS$ to a ``nice'' job set $\uwJS$
(to be defined) and show that
such a transformation only increases the cost modestly.
Furthermore, we show that for any nice job set $\uwJS$, we can bound $\cost(\Alg(\uwJS))$ by $\cost(\Opt(\uwJS))$ and in turn by $\cost(\Opt(\JS))$. 
Then we can establish the competitive ratio of~$\Alg$.

\subsection{Nice job set and transformations}
\label{sec:general_nice}

A job $J$ is said to be a \emph{nice job} if $\w(J) = 2^p$, for some non-negative integer $p$
and 
a job set $\uwJS$ is said to be a \emph{nice job set} if all its jobs are nice jobs.
In other words, the nice job $J$ is in class $\class_{p}$.

\runtitle{Procedure $\convertw$.}
Given a job set $\JS$, we define the procedure $\convertw$ to transform each job $J\in \JS$ into a nice job $\uwJ$ as follows. 
We denote the resulting nice job set by $\uwJS$. Suppose $J$ is in class $\class_{p}$.
We modify its width, release time and deadline.
\begin{itemize}
\item $\w(\uwJ) \leftarrow 2^p$;
\item $\rel(\uwJ) \leftarrow \rel(J)$;
\item $\dl(\uwJ) \leftarrow \rel(\uwJ) + \max\{\dl(J)-\rel(J), 2^p\}$.
\end{itemize}


Modifications to $\rel(\uwJ)$ and $\dl(\uwJ)$ are due to rounding up the width. The observation below follows directly from the definition.

\begin{observation}
\label{thm:nice_interval}
For any job $J$ and its nice job $\uwJ$ transformed by $\convertw$,
\text{\rm (i)} $\interval(J) \subseteq \interval(\uwJ)$;
\text{\rm (ii)} $\interval(J) \not= \interval(\uwJ)$ if and only if $|\interval(J)| < 2^p$;
  in this case, $\den(J) > \frac{1}{2}$ and $\den(\uwJ) = 1$.
\end{observation}

We then define two procedures that transform schedules related to nice job sets.
$\relaxS$ takes a schedule $\sch$ for a job set $\JS$ and converts it to a schedule $\uwSch$ for the corresponding nice job set $\uwJS$.
On the other hand, $\shrinkS$ takes a schedule $\uwSch$ for a nice job set $\uwJS$ and converts it to a schedule $\sch$ for $\JS$. 

\runtitle{Transformation $\relaxS$.} $\relaxS$ transforms $\sch$ into $\uwSch$ by moving the start and end time of every job $J$.

\begin{itemize}
\item $\stime(\uwSch,\uwJ) = \min\{ \dl(\uwJ)-\w(\uwJ), \stime(S,J)\}$
\item $\etime(\uwSch,\uwJ) = \stime(\uwSch,\uwJ) + \w(\uwJ)$.
\end{itemize}

Observation~\ref{thm:relax_feasible} asserts that the resulting schedule $\uwSch$ is feasible for $\uwJS$ while Lemmas~\ref{thm:relax_load} and~\ref{thm:relax_cost} analyze the load and cost of the schedule.

\begin{observation}
\label{thm:relax_feasible}
Consider any schedule $\sch$ for $\JS$ and the schedule $\uwSch$ constructed by $\relaxS$ for the corresponding $\uwJS$.
We have $[\stime(\uwSch, \uwJ), \etime(\uwSch,\uwJ)] 
\subseteq [\rel(\uwJ), \dl(\uwJ)]$; in other words, $\uwSch$ is a feasible schedule for $\uwJS$.
\end{observation}

To analyze the load of the schedule $\uwSch$, we consider
partial schedule $\uwSchp \subseteq \uwSch$ (resp.\ $\schp \subseteq \sch$)
which is for all the jobs of $\uwJS$ (resp.\ $\JS$) in class $\class_p$.
Intuitively, the load of $\uwSchp$ at any time is at most
the sum of the load of $\schp$ at the current time and $2^{p-1}-1$ timeslots
before and after the current time.

\begin{lemma}
\label{thm:relax_load}
At any time~$t$, $\load(\uwSchp, t) \leq \load(\schp, t) + \load(\schp, t-(2^{p-1}-1)) + \load(\schp, t+(2^{p-1}-1)).$
\end{lemma}

\begin{proof}
We prove that for any job $J$, $\uwJ$ contributes to $\load(\uwSchp, t)$ only if $J$ contributes to either $\load(\schp,t)$, $\load(\schp, t-(2^{p-1}-1))$, or $\load(\schp, t+(2^{p-1}-1))$, . 
There are two cases that $J$ does not contribute to $\load(\schpq, t-(2^{p-1}-1))$ nor $\load(\schpq, t+(2^{p-1}-1))$: (i) $\etime(J) < t-(2^{p-1}-1)$ or $\stime(J)> t+ (2^{p-1}-1)$, and (ii) $[\stime(J), \etime(J)] \subseteq (t-(2^{p-1}-1), t+(2^{p-1}-1))$.

Consider case (i).  $\etime(\uwJ) \leq \etime(J) + (2^{p-1}-1)$ and $\stime(\uwJ) \geq \stime(J) - (2^{p-1}-1)$. 
Hence, $t\notin [\stime(\uwJ), \etime(\uwJ)]$ if $\etime(J)< t-(2^{p-1}-1)$ or $\stime(J) > t+(2^{p-1}-1)$. 
That is, $\uwJ$ does not contribute to $\load(\uwSchp, t)$. Notice that if $\etime(J) = t-(2^{p-1}-1)$ or $\stime(J) = t+(2^{p-1}-1)$, $J$ does not necessarily contribute to $\load(\uwSchp, t)$. We count the contribution for worst case analysis.


For case (ii), consider job $J$ with $[\stime(J), \etime(J)] \subseteq (t-(2^{p-1}-1), t+(2^{p-1}-1))$. Since $2^{p-1}<\w(J)\leq 2^p$, $t\in [\stime(J), \etime(J)]$. That is, $J$ contributes to $\load(\schp, t)$ no matter if $\uwJ$ contributes to $\load(\uwSchp, t-(2^{p-1}-1))$ or $\load(\uwSchp, t+(2^{p-1}-1))$.

By case (i) and (ii), for any job $J$ with $[\stime(J), \etime(J)]\cap [t-(2^{p-1}-1), t+(2^{p-1}-1)]=\emptyset$, $\uwJ$ does not contribute to $\load(\uwSchp,t)$. And for any job $J$ with $[\stime(J), \etime(J)]\subseteq (t-(2^{p-1}-1), t+(2^{p-1}-1))$, $J$ contributes to $\load(\schp,t)$. Hence, by assuming all jobs at timeslot $t-(2^{p-1}-1)$ or $t+(2^{p-1}-1)$ contribute to $\load(\uwSchp, t)$, $\load(\uwSchp, t)$ is bounded by $\load(\schp, t) + \load(\schp, t-(2^{p-1}-1)) + \load(\schp, t+(2^{p-1}-1))$.
\end{proof}

\begin{lemma}
\label{thm:relax_cost}
Using $\relaxS$, we have 
$\cost(\uwSchp) \leq 3^\alpha\cdot \cost(\schp)$.
\end{lemma}

\begin{proof}
By Lemma~\ref{thm:relax_load}, $\cost(\uwSchp) = \sum_t \load(\uwSchp,t)^\alpha \leq \sum_t (\load(\schp, t) + \load(\schp, t-(2^{p-1}-1)) + \load(\schp, t+(2^{p-1}-1)))^\alpha\leq \sum_t(3\cdot\load(\schp, t))^\alpha = 3^\alpha\cdot\cost(\schp)$.
%
%
\end{proof}

\runtitle{Transformation $\shrinkS$.}
On the other hand, $\shrinkS$ converts a schedule $\uwSch$ for a nice job set $\uwJS$ to a schedule $\sch$ for the corresponding job set $\JS$.
We set
\begin{itemize}
\item $\stime(\sch,J) \leftarrow \stime(\uwSch,\uwJ)$;
\item $\etime(\sch,J) \leftarrow \stime(\sch,J) + \w(J) $, 
therefore, $\etime(\sch,J) \leq \etime(\uwSch,\uwJ)$.
\end{itemize}

Observation~\ref{thm:shrink_feasible} asserts that the resulting schedule $\sch$ is feasible for $J$
and Lemma~\ref{thm:shrink_cost} analyzes the cost of the schedule.

\begin{observation}
\label{thm:shrink_feasible}
Consider any schedule $\uwSch$ for $\uwJS$ and schedule $\sch$ constructed by $\shrinkS$ for the corresponding $\JS$. 
For any $\uwJ$ and the corresponding $J$, we have
\text{\rm (i)} $[\stime(\sch,J), \etime(\sch,J)] \subseteq [\stime(\uwSch,\uwJ), \etime(\uwSch,\uwJ)]$;
\text{\rm (ii)} $[\stime(\sch,J), \etime(\sch,J)] \subseteq [\rel(J),\dl(J)]$.
\end{observation}

By Observation~\ref{thm:shrink_feasible}, we have the following lemma.
\begin{lemma}
\label{thm:shrink_cost}
Using $\shrinkS$, we have $\cost(\schp) \leq \cost(\uwSchp)$.
\end{lemma}

\subsection{The online algorithm}
\label{sec:general_algorithm}
\runtitle{Online algorithm $\Alg$.}
We are now ready to describe the algorithm $\Alg$ for an arbitrary job set $\JS$.
When a job $J$ is released, it is converted to $\uwJ$ by $\convertw$ and classified into one of the classes $\class_{p}$.
Jobs in the same class after $\convertw$ (being a uniform-width job set) are scheduled by $\AlgUV$ independently of other classes.
We then modify the execution time of $\uwJ$ in $\AlgUV$ to the execution time of $J$ in $\Alg$ by Transformation $\shrinkS$.
Note that all these procedures can be done in an online fashion.

Using the results in Sections~\ref{sec:general_unit_width} and~\ref{sec:general_nice}, we can compare the cost of $\Alg(\JS)$ with $\Opt(\uwJSp)$ for each class $\class_{p}$ (see Theorem~\ref{thm:general}). 
It remains to analyze the cost of $\Opt(\uwJSp)$ and $\Opt(J)$ in the next observation.

\begin{observation}
\label{thm:opt_nice}
Consider any job set $\JS$, its corresponding job set $\uwJS$ and the corresponding job set of each class $\JSp$ and $\uwJSp$.
\text{\rm (i)} $\cost(\Opt(\uwJSp)) \leq 3^\alpha \cdot \cost(\Opt(\JSp))$;
\text{\rm (ii)} $\cost(\Opt(\JSp)) \leq \cost(\Opt(\JS))$.
\end{observation}

\begin{proof}
%
(i) Given $\Opt(\JSp)$, there exists schedule $S(\uwJSp)$ generated by $\relaxS$. By Lemma~\ref{thm:relax_cost} , 
$\cost(S(\uwJSp))\leq 3^\alpha\cdot\cost(\Opt(\JSp))$. Hence, $\cost(\Opt(\uwJSp))\leq\cost(S(\uwJSp))\leq 3^\alpha\cdot\cost(\Opt(\JSp))$.

(ii) Assume on the contrary that $\cost(\Opt(\JS)) < \cost(\Opt(\JSp))$, we can generate a schedule $S(\JSp)$ by removing jobs from $\Opt(\JS)$ which are not in $\JSp$. It follows that $\cost(S(\JSp)) \leq \cost(\Opt(\JS)) < \cost(\Opt(\JSp))$, contradicting to the fact that $\Opt(\JSp)$ is optimal for $\JSp$.
\end{proof}

\begin{theorem}
\label{thm:general}
For any job set $\JS$, we have
$\cost(\Alg(\JS)) \leq (36 \ceilLogKw)^\alpha \cdot \left( 8(e+e^2)^\alpha+1\right)\cdot \cost(\Opt(\JS))$, where $\Kw=\frac{\wmax}{\wmin}$.
\end{theorem}

\begin{proof}
By definition, $\cost(\Alg(\JS)) = \sum_t \load(\Alg(\JS), t)^\alpha = \sum_t(\sum_{p=1}^{\ceilLogKw} \load(\Alg(\JSp),t))^\alpha$.
The latter is at most $\ceilLogKw^{\alpha-1}\sum_{p=1}^{\ceilLogKw}\sum_t\load(\Alg(\JSp),t)^\alpha$. For each group of jobs $\JSp$, we $\convertw$ it to $\uwJSp$, apply algorithm $\AlgUV$ on it, and transform the schedule into a schedule for $\JSp$ by $\shrinkS$. Hence, $\load(\Alg(\JSp),t)\leq\load(\AlgUV(\uwJSp),t)$ for each $t$. It follows that $\cost(\Alg(\JS)) \leq \ceilLogKw^{\alpha-1}\sum_{p=1}^{\ceilLogKw}\cost(\Alg(\JSp))\leq \ceilLogKw^{\alpha-1}\sum_{p=1}^{\ceilLogKw}\cost(\AlgUV(\uwJSp))$. 
By Lemma~\ref{thm:cost_UV} and Observations~\ref{thm:opt_nice} (i) and \ref{obs:subopt}, $\cost(\AlgUV(\uwJSp))\leq 12^\alpha \cdot (8 (e+e^2)^\alpha+1)\cdot\cost(\Opt(\uwJSp))\leq 12^\alpha \cdot (8 (e+e^2)^\alpha+1)\cdot 3^\alpha\cdot\cost(\Opt(\JSp))\leq 36^\alpha \cdot (8(e+e^2)^\alpha+1)\cdot\cost(\Opt(\JS))$. 
Hence $\cost(\Alg(\JS))\leq 36^\alpha\cdot\ceilLogKw^{\alpha-1}\cdot( 8(e+e^2)^\alpha+1)\cdot\sum_{p=1}^{\ceilLogKw}\cost(\Opt(\JS))= (36\ceilLogKw)^\alpha\cdot(8(e+e^2)^\alpha+1)\cdot \cost(\Opt(\JS))$.
\end{proof}

Note that the logarithm in the competitive ratio comes from the number of classes defined in Section~\ref{sec:prelim}.
Suppose we change the definition of classes such that class $p$ includes jobs of size in the range $((1+\lambda)^{p-1}, (1+\lambda)^p]$
for some $\lambda > 0$ and Procedure $\convertw$ such that the width of jobs in class $\class_{p}$ is round up to $(1+\lambda)^p$.
Then, the number of classes becomes $\lceil \log_{1+\lambda} \Kw \rceil$.
In addition, the competitive ratio depends on Lemma~\ref{thm:relax_load} 
that bounds the load at any timeslot by the load of three other timeslots.
This number of timeslots is also affected by the definition of classes.
In summary, the following lemma states the competitive ratio for varying $\lambda$.

\begin{lemma}
\label{thm:lambda}
For $0 < \lambda \leq 0.5$, $0.5 < \lambda \leq 1$ and $\lambda > 1$,
the competitive ratio of our algorithm becomes
$(12 \times 2 \lceil \log_{1+\lambda} \Kw \rceil)^\alpha (8(e+e^2)^\alpha+1)$,
$(12 \times 3 \lceil \log_{1+\lambda} \Kw \rceil)^\alpha (8(e+e^2)^\alpha+1)$, and
$(12 \times (2\lambda+1) \lceil \log_{1+\lambda} \Kw \rceil)^\alpha (8(e+e^2)^\alpha+1)$,
respectively.
\end{lemma}

\begin{proof}
The number of classes is $\lceil \log_{1+\lambda} \Kw \rceil$, which replaces $\ceilLogKw$ in Theorem~\ref{thm:general}.
We note that this number decreases as $\lambda$ increases.
In Lemma~\ref{thm:relax_load},
the load of $\uwSchp$ at any time $t$ is bounded by the load of $\schp$ at three timeslots when $\lambda=1$.
We observe that this property stays the same for $0.5 < \lambda \leq 1$.
Using a similar argument, we can show that
if $\lambda$ is smaller and $0 < \lambda \leq 0.5$, then the number of timeslots involved becomes smaller and equals to $2$.
Furthermore, when $\lambda > 1$, the number of timeslots increases and equals to $2\lambda+1$.
\end{proof}

We note the competitive ratio for $\lambda<1$ is larger than that for $\lambda=1$,
and the best competitive ratio occurs when $1 < \lambda < 2$.

%% file: general_lb.tex
\renewcommand{\Alg}{{\mathcal{A}}}

\subsection{Lower bound}
\label{sec:general_lb}

In this section, we show lower bounds on competitive ratio for Grid problem with unit height and arbitrary width by designing an adversary for the problem. 
The lower bounds are immediately lower bounds for the general case of Grid problem.

The adversary constructs a set of jobs with a low cost of offline optimal schedule but a high cost of any online algorithm~$\Alg$.
It generates jobs one by one and assigns release times, deadlines and widths of jobs based on the previously generated jobs.
The start times of jobs scheduled by $\Alg$ will be used for the job generations later.
This ensures that $\Alg$ has to put a job on top of all existing jobs and results in a high energy cost.
Meanwhile, the adversary will choose an appropriate feasible interval for each job such that an optimal offline algorithm can schedule the job set with low energy cost.
The following is the description of the adversary.

\runtitle{Adversary~$\adv$ and job instance~$\Ja$.}
Given an online algorithm~$\Alg$, a constant~$\alpha > 1$ and a large number~$x$, adversary~$\adv$ outputs a set of jobs~$\Ja$ consisting of $\lfloor \alpha \rfloor + 1$ jobs.
Let $J_i$ be the $i$th job of $\Ja$.
The adversary first computes a width for each job before running $\Alg$.
It sets $\w(J_{\lfloor \alpha \rfloor}) = x$, $\w(J_{\lfloor \alpha \rfloor + 1}) = x - 1$, and $\w(J_i) = 3\w(J_{i+1}) + 1$ for $1 \leq i \leq \lfloor \alpha \rfloor - 1$.
Then $\adv$ releases the jobs from $J_1$ to $J_{\lfloor \alpha \rfloor + 1}$ accordingly and computes a release time and a deadline for each job through an interaction with $\Alg$.
For the first job~$J_1$, $\adv$ chooses any release time and deadline such that $\dl(J_1) - \rel(J_1) \geq 3 \w(J_1)$.
For the $i$th job~$J_i \in \Ja$ for $2 \leq i \leq \lfloor \alpha \rfloor + 1$ accordingly, $\adv$ sets $\rel(J_i) = \stime(\Alg, J_{i-1}) + 1$ and $\dl(J_i) = \etime(\Alg, J_{i-1})$.
This limits $\Alg$ to fewer choices of start times for scheduling a new job.
A job can only be scheduled in the execution interval of its previous job by $\Alg$.
On the other hand, no two jobs have the same release time.

Let $\wmax$ and $\wmin$ denote by the maximum and minimum width of jobs respectively, and let $\Opt$ be an optimal offline algorithm for Grid problem.
We have the following results.

\begin{lemma}
$\cost(\Opt(\Ja)) \leq x \cdot 3^{\lfloor \alpha \rfloor}$.
\label{lm:cost_opt_ja}
\end{lemma}
\begin{proof}
By the setting of $\adv$, we show that $\Opt$ can schedule all the jobs  $\Ja$ without overlapping, and the cost of an optimal schedule is just the sum of widths of all the jobs.

For any job $J_i \in \Ja$ and $i \geq 2$, the length of its feasible interval is $\dl(J_i) - \rel(J_i) = \etime(\Alg, J_{i-1}) - (\stime(\Alg, J_{i-1}) + 1) = \w(J_{i-1}) - 1 = 3\w(J_i)$.
This means no matter where we schedule a job, at least one of the intervals $[\rel(J_i), \stime(J_i))$ and $[\etime(J_i), \dl(J_i))$ has length at least $\w(J_i)$.
Algorithm~$\Opt$ can schedule the subsequent jobs in the interval with length at least $\w(J_i)$
such that the subsequent jobs do not overlap with $J_i$.
This is because the sum of widths of all the subsequent jobs does not exceed $\w(J_i)$.
Since this argument can be applied on all the jobs, this implies that all the jobs do not overlap with each other in an optimal schedule.
Thus the cost of an optimal schedule is the sum of widths of all the jobs.
More precisely,
\begin{align*}
\cost(\Opt(\Ja)) &= (x - 1) + x + (3x + 1) + (3(3x + 1) + 1) + \ldots + \wmax \\
&\leq 2x + 2 \cdot 3x + 2 \cdot 9x + \ldots + 2 \cdot 3^{\lfloor \alpha \rfloor - 1} x \\
&= 2x \cdot \frac{3^{\lfloor \alpha \rfloor} - 1}{2} \leq x \cdot 3^{\lfloor \alpha \rfloor}
\enspace .
\qedhere
\end{align*}
\end{proof}

\begin{theorem}
For any deterministic online algorithm $\Alg$ for Grid problem with unit height and arbitrary width, adversary $\adv$ constructs an instance $\Ja$ such that
\text{\rm (i)} for constant $\alpha$,
\[ \frac{\cost(\Alg(\Ja))}{\cost(\Opt(\Ja))} \geq \left(\frac{\lfloor \alpha \rfloor + 1}{3}\right)^\alpha \enspace ;\]
and
\text{\rm (ii)} for arbitrary $\alpha$,
\[ \frac{\cost(\Alg(\Ja))}{\cost(\Opt(\Ja))} \geq \left(\frac{1}{3} \log \frac{\wmax}{\wmin}\right)^\alpha \enspace .\]
\end{theorem}
\begin{proof}
We first give a lower bound on $\cost(\Alg(\Ja))$ and then give the lower bounds on the competitive ratio by combining $\cost(\Alg(\Ja))$ with Lemma~\ref{lm:cost_opt_ja}.

(i)
By the setting of $\adv$, all the jobs scheduled by $\Alg$ overlap with each other.
For ease of the computation for the cost of $\Alg$, we only consider the timeslots contained by the execution interval of the last job $J_{\lfloor \alpha \rfloor + 1}$.
Thus $\cost(\Alg(\Ja)) \geq (x - 1) \cdot (\lfloor \alpha \rfloor + 1)^\alpha$ 
and
\begin{align*}
\frac{\cost(\Alg(\Ja))}{\cost(\Opt(\Ja))}
\geq
\frac{(x - 1) \cdot (\lfloor \alpha \rfloor + 1)^\alpha}{x \cdot 3^{\lfloor \alpha \rfloor}}
\geq
\left( \frac{\lfloor \alpha \rfloor + 1}{3} \right)^\alpha
\end{align*}
as $x$ to be large enough.

(ii)
Assume $\alpha$ can be arbitrarily large.
We use $\wmax$ and $\wmin$ to bound $\lfloor \alpha \rfloor + 1$.
According to Lemma~\ref{lm:cost_opt_ja}, we have $\wmax \leq \cost(\Opt(\Ja)) \leq x \cdot 3^{\lfloor \alpha \rfloor}$, and thus
\begin{align*}
\lfloor \alpha \rfloor &\geq \log_3 \frac{\wmax}{x} \geq \log_3 \frac{\wmax}{3(x - 1)} = \log_3 \frac{\wmax}{\wmin} - 1
\enspace .
\end{align*}
Note that $x \leq 3(x - 1)$ if $x \geq 2$.
Therefore, $\cost(\Alg(\Ja)) \geq (x - 1) \cdot \log^\alpha \frac{\wmax}{\wmin}$.
Combining with Lemma~\ref{lm:cost_opt_ja}, we have the lower bound on the competitive ratio
\begin{align*}
\frac{\cost(\Alg(\Ja))}{\cost(\Opt(\Ja))} &\geq \frac{(x - 1) \cdot \log^\alpha \frac{\wmax}{\wmin}}{x \cdot 3^{\lfloor \alpha \rfloor}} \geq \left(\frac{1}{3} \log \frac{\wmax}{\wmin}\right)^\alpha
\end{align*}
as $x$ to be large enough.
\end{proof}

\begin{corollary}
For any deterministic online algorithm for Grid problem, the competitive ratio is at least
\text{\rm (i)}
$( \frac{\lfloor \alpha \rfloor + 1}{3})^\alpha$
for constant $\alpha$; and
\text{\rm (ii)}
$(\frac{1}{3}\log \frac{\wmax}{\wmin})^\alpha$
for arbitrary $\alpha$.
\end{corollary}

%% file: horizontalUniform.tex
\section{Online algorithm for uniform height jobs}
\label{sec:special}


In this section we focus on uniform-height jobs of height $h$ and consider two special cases of the width.
We first consider jobs with uniform-height and unit-width (Section~\ref{sec:special_unit})
and secondly consider jobs with agreeable deadlines (Section~\ref{sec:horizontal_agreeable}).

To ease the discussion, we refine a notation we defined before.
For any algorithm~$\Alg$ for a job set $\JS$ and a time interval $\I$,
we denote by $\Alg(\JS,\I)$ the schedule of $\Alg$ on $\JS$ over the time interval $\I$.

\subsection{Main ideas}
\label{sec:horizontal_framework}


The main idea is to make reference to the online algorithm $\AVR$ and consider two types of intervals, $\IL{h}$ where the average load is higher than~$h$ and $\IS{h}$ where the average load is at most $h$.
For the former, we show that we can base on the competitive ratio of $\AVR$ directly; for the latter, our load could be much higher than that of $\AVR$ and in such case, we compare directly to the optimal algorithm.
Combining the two cases, we have Lemma~\ref{thm:horizontal_framework},
which holds for any job set. 
In Sections~\ref{sec:special_unit} and~\ref{sec:horizontal_agreeable}, 
we show how we can use this lemma to obtain algorithms for the special cases.
Notice that the number $\lceil\frac{\avg(t)}{h}\rceil$ is the minimum number of jobs needed to make the load at $t$ at least $\avg(t)$.

\begin{lemma}
\label{thm:horizontal_framework}
Suppose we have an algorithm $\Alg$ for a any job set $\JS$ such that for some $c$ and $c'$
\text{\rm (i)} $\load(\Alg,t) \leq c\cdot \h\cdot \lceil\frac{\avg(t)}{\h}\rceil$ for all $t\in \IL{\h}$, 
and \text{\rm (ii)} $\load(\Alg,t) \leq c'\cdot \h$ for all $t\in \IS{h}$. 
Then we have $\cost(\Alg(\JS))\leq (\frac{(4c\alpha)^\alpha}{2}+c'^\alpha)\cdot cost(\Opt(\JS))$.
\end{lemma}

\begin{proof}

We denote the speed of $\AVR$ at $t$ as $\load(\AVR, t)$. 
We are going to prove that (a) $\cost(\Alg(\JS, \IL{\h})) \leq \frac{(4c\alpha)^\alpha}{2} \cdot \cost(\Opt(\JS))$ and (b) $\cost(\Alg(\JS, \IS{\h})) \leq c'^\alpha \cdot \cost(\Opt(\JS))$. Hence, the total cost $\cost(\Alg(\JS)) \leq (\frac{(4c\alpha)^\alpha}{2} + c'^\alpha) \cdot  \cost(\Opt(\JS))$ since $\IL{h}$ and $\IS{h}$ are disjoint.

(a)
We compare $\load(\Alg, t)$ to $\load(\AVR, t)$ for each timeslot~$t$ in $\IL{\h}$.
The assumption of $\Alg$ means that $\load(\Alg, t) \leq c\cdot h\cdot \lceil \frac{\avg(t)}{h}\rceil <  c\cdot h \cdot ( \frac{\avg(t)}{h} +1) = c\cdot (\avg(t) + h)\leq 2c\cdot\avg(t)$ since $\avg(t) >h$. 
By definition, $\cost(\Alg(\JS, \IL{\h})) = \sum_{t\in\IL{\h}} \load(\Alg, t)^\alpha \leq \sum_{t\in\IL{\h}} (2c\cdot\avg(t))^\alpha$. 
Recall that $\load(\AVR, t)  = \avg(t)$ for each~$t$. 
Hence, by Corollary~\ref{thm:AVR},
$\cost(\Alg(\JS, \IL{\h}))\leq (2c)^\alpha\cdot\cost(\AVR(\JS,\IL{\h})) \leq (2c)^\alpha\cdot\cost(\AVR(\JS)) \leq \frac{(4c\alpha)^\alpha}{2}\cdot\cost(\Opt(\JS))$.

(b)
Since only jobs which are available in $\IS{\h}$ can be scheduled at $t \in \IS{\h}$, $\cost(\Alg(\JS,\IS{\h})) \leq \sum_{\Job{}: \interval{} \cap \IS{\h}\neq \emptyset} \Width{} \cdot (c'\h)^\alpha \leq \sum_{\Job{}\in \JS} \Width{} \cdot (c'\h)^\alpha$. 
By convexity, $\cost(\Opt(\JS)) \geq \sum_{\Job{}\in \JS} \Width{}\cdot \Height{}^\alpha = \sum_{\Job{}\in \JS} \Width{}\cdot \h^\alpha$. 
Hence, $\cost(\Alg(\JS,\IS{\h})) \leq c'^\alpha \cdot \cost(\Opt(\JS))$.


Adding up the two cost,
we have $\cost(\Alg(\JS)) 
= \cost(\Alg(\JS, \IL{h})) + \cost(\Alg(\JS, \IS{h})) 
\leq (\frac{(4c\alpha)^\alpha}{2}+c'^\alpha)\cdot\cost(\Opt(\JS))$
and the theorem follows.
\end{proof}

\subsection{Uniform-height and unit-width}
\label{sec:special_unit}

In this section we consider job sets where all jobs have uniform-height and unit-width, i.e., $\w(J) = 1$ and $\h(J) = h$  for all jobs $J$. 
Note that such case is a subcase discussed in Section~\ref{sec:vertical}.
Here we illustrate a different approach using the ideas above
and describe the algorithm $\AlgOne$ for this case.
The competitive ratio of $\AlgOne$ is better than that of Algorithm $\AlgV$ in Section~\ref{sec:vertical} when $\alpha < 3.22$.

\runtitle{Algorithm $\AlgOne$.}
At any time $t$, 
choose $\lceil \frac{\avg(t)}{h}\rceil$ jobs according to the EDF rule and schedule them to start at $t$. 
If there are fewer jobs available, schedule all available jobs.

The next theorem asserts that the algorithm gives feasible schedule 
and states its competitive ratio.

\begin{theorem}
\label{thm:alg_unit}
\text{\rm (i)} The schedule constructed by Algorithm~$\AlgOne$ is feasible.
\text{\rm (ii)} Algorithm~$\AlgOne$ is $(\frac{(4\alpha)^\alpha}{2}+1)$-competitive.
\end{theorem}

\begin{proof}
%
(i) The feasibility can be proved by comparing to $\AVR$.
At any time $t$, the total work done by $\AVR$ in interval $[0, t)$ is $\sum_{t'<t}\avg(t')$. 
On the other hand, the total work done by $\AlgOne$ in the same interval is $\sum_{t'<t} h\cdot\lceil\frac{\avg(t)}{h}\rceil \geq \sum_{t'<t} h\cdot\frac{\avg(t)}{h} = \sum_{t'<t} \avg(t)$ if there are enough available jobs in this interval. 
If the number of available jobs is less than  $\sum_{t'<t} h\cdot\lceil\frac{\avg(t)}{h}\rceil $, $\AlgOne$ will execute all these jobs. 
The work done by $\AlgOne$ within interval $[0,t)$ is at least the work done by $\AVR$ in both cases. 
Hence, $\AlgOne$ is feasible since $\AVR$ is feasible.

(ii) We note that $\load(\AlgOne,t) \leq h\cdot\lceil\frac{\avg(t)}{h}\rceil $.
To use Lemma~\ref{thm:horizontal_framework}, we can set $c'=1$ for $t\in\IS{h}$ by the definition of $\IS{h}$.
Furthermore, we can set $c=1$ for $t\in\IL{h}$.
\end{proof}

\subsection{Uniform-height, arbitrary width and agreeable deadlines}
\label{sec:horizontal_agreeable}

In this section we consider jobs with agreeable deadlines.
We first note that simply scheduling $\lceil \frac{\avg(t)}{h} \rceil$ number of jobs may not return a feasible schedule.

\begin{example}
Consider four jobs each job $J$ with $\rel(J)=0$, $\dl(J)=5$, $\h(J)=h$, $\w(J)=3$.
Note that $\avg(t) = 2.4\cdot h$ for all $t$.
If we schedule at most $\lceil \frac{\avg(t)}{h} \rceil = 3$ jobs at any time, we can complete three jobs but the remaining job cannot be completed.
To schedule all jobs feasibly, we need at least two timeslots where all jobs are being executed.
\end{example}

To schedule these jobs, we first observe in Lemma~\ref{thm:horizontal_density} that for a set of jobs with total densities at most~$h$, it is feasible to schedule them such that the load at any time is at most~$h$.
Roughly speaking, we consider jobs in the order of release, and hence, in EDF manner since the jobs have agreeable deadlines.
We keep the current ending time of all jobs that have been considered.
As a new job is released, if its release time is earlier than the current ending time, we set its start time to the current ending time (and increase the current ending time by the width of the new job); otherwise, we set its start time to be its release time.
Lemma~\ref{thm:horizontal_density} asserts that such scheduling is feasible
and maintains the load at any time to be at most~$h$.

Using this observation, we then partition the jobs into ``queues'' each of which has sum of densities at most~$h$.
Each queue $\queue_i$ is scheduled independently and the resulting schedule is to ``stack up'' all these schedules.
The queues are formed in a Next-Fit manner:
(i) the current queue $\queue_q$ is kept ``open'' and a newly arrived job is added to the current queue if including it makes the total densities stays at most $1$;
(ii) otherwise, the current queue is ``closed'' and a new queue $\queue_{q+1}$ is created as open.

\begin{lemma}
\label{thm:horizontal_density}
Given any set of jobs of uniform-height, arbitrary-width and agreeable deadlines.
If the sum of densities of all these jobs is at most~$h$, then it is feasible to schedule all of them using a maximum load~$h$ at any time. That is, there is no stacking up among these jobs.
\end{lemma}

\begin{proof}
Suppose there are $k$ jobs $J_1, J_2, \cdots, J_k$ such that $\sum_{1 \leq i \leq k} \den(J_i)\leq h$. 
Without loss of generality, we assume that $\dl(J_i) \leq \dl(J_j)$ and $\rel(J_i)\leq \rel(J_j)$ for $1\leq i < j \leq k$. 
We claim that it is feasible to set $[\stime(J_i),\etime(J_i))$ to $[\max\{\rel(J_i), \etime(J_{i-1})\}, \stime(J_i)+\w(J_i))$ for all $ 1<i\leq k$ and $[\stime(J_1), \etime(J_1)) = [\rel(J_1), \rel(J_1)+\w(J_1))$. 

We observe that $\den(J_1)\leq h$ since $\sum_i \den(J_i) \leq h$. 
It is feasible to set $[\stime(J_1),\etime(J_1))$ to $[\rel(J_1), \rel(J_1)+\w(J_1))$ since the input is feasible. 
Then we have to prove that 
$[\stime(J_i), \etime(J_i)) = [\max\{\rel(J_i), \etime(J_{i-1})\}, \stime(J_i)+\w(J_i)) \subseteq [\rel(J_i), \dl(J_i))$. 
Since $\stime(J_i) = \max\{\rel(J_i), \etime(J_{i-1})\}$, we have $\stime(J_i)\geq \rel(J_i)$. 
Assume that $\cup_{g\leq i}\interval(J_g)$ is a contiguous interval. 
Since $\sum_{g\leq i}\den(J_g)\leq h$, $\sum_{g\leq i} \frac{\w(J_g)}{\dl(J_g)-\rel(J_g)} \leq 1$. 
From the ordering of jobs we have $1\geq \sum_{g\leq i}\frac{\w(J_g)}{\dl(J_g)-\rel(J_g)}\geq \sum_{g\leq i}\frac{\w(J_g)}{\dl(J_i)-\rel(J_1)}$. 
Hence, $\sum_{g\leq i}\w(J_g) \leq \dl(J_i)-\rel(J_1)$. Therefore $J_i$ can be finished before $\dl(J_i)$. 
On the other hand, if $\cup_{g\leq i}\interval(J_g)$ is not contiguous. 
The proof above shows that for each contiguous, each of the involving jobs can be finished by its deadline.
\end{proof}

\runtitle{Algorithm $\AlgAgree$.}
The algorithm consists of the following components: InsertQueue, SetStartTime and ScheduleQueue.

\noindent
\textbf{\emph{InsertQueue:}}
We keep a counter $q$ for the number of queues created.
When a job~$J$ arrives, if $\den(J) + \sum_{J' \in \queue_q} \den(J') \leq h$, then job $J$ is added to $\queue_q$; 
otherwise, job $J$ is added to a new queue $\queue_{q+1}$ and we set $q \leftarrow q+1$.

\noindent
\textbf{\emph{SetStartTime:}}
For the current queue, we keep a current ending time $E$, initially set to $0$.
When a new job $J$ is added to the queue, if $\rel(J) \leq E$,
we set $\stime(J) \leftarrow E$;
otherwise, we set $\stime(J) \leftarrow \rel(J)$.
We then update $E$ to $\stime(J)+\w(J)$.

\noindent
\textbf{\emph{ScheduleQueue:}}
At any time~$t$, schedule all jobs in all queues with start time set at $t$.

By Lemma~\ref{thm:horizontal_density},
the schedule returns by $\AlgAgree$ is feasible.
We then analyze its load and hence, derive its competitive ratio.
Recall the definition of $\IL{h}$ and $\IS{h}$

\begin{lemma}
\label{thm:agree_load}
Using $\AlgAgree$, we have
\text{\rm (i)} $\load(\AlgAgree,t)\leq 2\cdot h\cdot\lceil\frac{\avg(t)}{h}\rceil$ for $t\in\IL{h}$; 
\text{\rm (ii)} $\load(\AlgAgree,t)\leq 2h$ for $t\in\IS{h}$.
\end{lemma}

\begin{proof}
For timeslot $t$, suppose there are $k$ queues ($Q_1, Q_2, \cdots, Q_k$) which contains jobs available at $t$. 
According to our algorithm, $\load(\AlgAgree,t)\leq k\cdot h$.

Let $D_i$ be the sum of densities of jobs in $Q_i$. 
Consider $t\in\IL{h}$. 
By the InsertQueue procedure, $D_i+D_{i+1}>h$ for $1 \leq i <k-1$. 
Therefore, if $k\geq 3$, $\avg(t) = \sum_{1\leq i\leq k}D_i> \sum_{1\leq i \leq k-1} D_i \geq \lfloor\frac{k-1}{2}\rfloor \cdot h$. 
It can be shown that $k\leq 2\cdot \lceil\frac{\avg(t)}{\h}\rceil$ since $\avg(t)>1$.\footnote{%
Let $r = \frac{\avg(t)}{\h}$, since $t\in \IL{\h}$, $r>1$. Since there are $k$ queues at $t$, the total density at $t$ is at least $\lceil \frac{k-1}{2} \rceil$. Hence, $\lceil\frac{k-1}{2}\rceil < r \leq \lceil r \rceil$.
Since $\lceil\frac{k-1}{2}\rceil$ is a natural number, $\lceil \frac{k-1}{2} \rceil \leq \lceil r \rceil-1$. 
It follows that $\lceil r \rceil -1 > \frac{k-1}{2}-1$. 
Therefore, $k \leq 2\lceil r \rceil = 2\lceil \frac{\avg(t)}{\h} \rceil$ since $k$ is a natural number.}
That is, $\load(\AlgAgree,t) = k\cdot h\leq 2\cdot h\cdot\lceil \frac{\avg(t)}{\h}\rceil$ for $t\in\IL{\h}$. 
On the other hand, if $k<3$, $\load(\AlgAgree, t) \leq 2\h < 2\cdot \avg(t) \leq 2\cdot \h \cdot \lceil \frac{\avg(t)}{\h} \rceil$.

For $t\in\IS{h}$, $\avg(t)\leq \h$ by definition. That is, the sum of densities of all available jobs at $t$ is no more than $\h$. 
By the InsertQueue procedure all jobs will be in at most two adjacent queues. Hence, $\load(\AlgAgree,t)\leq 2\h$ for $t\in\IS{\h}$.
\end{proof}

By Lemma~\ref{thm:agree_load} and Lemma~\ref{thm:horizontal_framework}, we have Theorem~\ref{thm:agree_cr} by setting $c=2$ and $c'=2$.
\begin{theorem}
\label{thm:agree_cr}
For jobs with uniform height, arbitrary width and agreeable deadlines, $\AlgAgree$ is $(\frac{(8\alpha)^\alpha}{2}+2^\alpha)$-competitive.
\end{theorem}


%% file: exact.tex
\section{Exact Algorithms}
\label{sec:exact}

In this section, we propose exact algorithms and derive lower bounds on the running time of exact algorithms.
Table~\ref{tab:fpt_summary} gives a summary of our results.

\begin{table}[t]
\begin{center}
\begin{tabular}{|c|c|c|}
\hline
\bf{Width} & \bf{Height} & \bf{Time complexity} \\ \hline \hline
Arbitrary & Arbitrary & $\wmax^{2m} \cdot (\win_{\max} + 1)^{4m} \cdot O(n^2)$ \\ \hline
Arbitrary & Arbitrary & $(4m \cdot \wmax^2)^{2m} \cdot O(n^2)$ \\ \hline
Unit & Arbitrary & $2^{O(N)}$ \\ \hline
\end{tabular}
\caption{Summary of our exact algorithms.}
\label{tab:fpt_summary}
\end{center}
\end{table}

\input{exact_fpt}

\input{exact_exp}

\input{exact_lb}

%% file: exact_fpt.tex
\subsection{Fixed parameter algorithms}
\label{sec:exact_fpt}

In parameterized complexity theory, the complexity of a problem is not only measured in terms of the input size, but also in terms of parameters.
The theory focuses on situations where the parameters can be assumed to be small, and the time complexity depends mainly on these small parameters.
The problems having such small parameters are captured by the concept ``fixed-parameter tractability''.
An algorithm with parameters $p_1, p_2, \ldots$ is said to be an \emph{fixed parameter algorithm} if it runs in $f(p_1, p_2, \ldots) \cdot O(g(N))$ time for any function $f$ and any polynomial function $g$, where $N$ is the size of input.
A parameterized problem is \emph{fixed-parameter tractable} if it can be solved by a fixed parameter algorithm.
In this section, we show that the general case of Grid problem, jobs with arbitrary release times, deadlines, width and height, is fixed-parameter tractable with respect to a few small parameters.

\subsubsection{Key notions}

We design two fixed parameter algorithms that are based on a dynamic programming fashion.
Roughly speaking, we divide the timeline into $k$ contiguous windows in a specific way, where each window~$\win_i$ represents a time interval $[\bound_i, \bound_{i+1})$ for $1 \leq i \leq k$.
The algorithm visits all windows accordingly from the left to the right and maintains a candidate set of schedules for the visited windows that no optimal solution is deleted from the set.
In the first fixed parameter algorithm, the parameters of the algorithm are the maximum width of jobs, the maximum number of overlapped feasible intervals and the maximum size of windows, where the latter two can be parameterized if we interpret the input job set as an ``interval graph''.
We will drop out the last parameter in the second algorithm.
All parameters do not increase necessarily as the number of jobs grows, and can be assumed to be small in practice.
For example, a width of a job is a requested amount of time to run an appliance, and the running time is usually a few hours, which is small when we make a timeslot to be an hour.
And the number of overlapped feasible intervals is at most the number of appliances.

\runtitle{Interval graph.}
A graph $G = (V,E)$ is an \emph{interval graph} if it captures the intersection relation for some set of intervals on the real line.
Formally, for each $v \in V$, we can associate $v$ to an interval $I_v$ such that $(u,v)$ is in $E$ if and only if $I_u \cap I_v \not= \emptyset$.
It has been shown in~\cite{fulkerson1965incidence,halin1982some} that an interval graph has a ``consecutive clique arrangement'', i.e., its maximal cliques can be linearly ordered in a way that for every vertex~$v$ in the graph, the maximal cliques containing $v$ occur consecutively in the linear order.
For any instance of the Grid problem, we can transform it into an interval graph $G=(V,E)$: For each job~$J$ with interval~$I(J)$, we create a vertex~$v(J) \in V$ and an edge is added between $v(J)$ and $v(J')$ if and only if
$I(J)$ intersects $I(J')$.
We can then obtain a set of maximal cliques in linear order, $C_1$, $C_2$, $\cdots$, $C_k$, by sweeping a vertical line from the left to the right, where $k$ denotes the number of maximal cliques thus obtained.
The parameter of our algorithm, the maximum number of overlapped feasible intervals, is just the maximum size of these maximal cliques.

\runtitle{Boundaries and windows.}
Based on the maximal cliques described above, we define some ``windows'' $\win_1$, $\win_2$, $\cdots$, $\win_k$ with ``boundaries'' $\bound_1$, $\bound_2$, $\cdots$, $\bound_{k+1}$ as follows.
We first give the definition of boundaries for the first algorithm.
This definition will be generalized in section~\ref{sec:exact_fpt2} for the second algorithm.
For $1 \leq i \leq k$, the $i$-th \emph{boundary} $\bound_i$ is defined as the earliest release time of jobs in clique $C_i$ but not in cliques before $C_i$, precisely, $b_i = \min \{ t \mid t = r(J) \text{ and } J \in  C_i\setminus(\cup_{s=1}^{i-1} C_s) \}$.
The rightmost boundary~$b_{k+1}$ is defined as the latest deadline among all jobs.
With the boundaries, we partition the timeslots into contiguous intervals called \emph{windows}.
The $i$-th window~$\win_i$ is defined as $[\bound_i, \bound_{i+1})$.

\runtitle{Example.}
Figure~\ref{fig:exampleIntervalGraph} is an example of a set of jobs, its corresponding interval graph and the corresponding maximal cliques.
The cliques are put in such a way that any vertex appears consecutively if there is two or more of it.
The boundaries of windows are determined by the leftmost vertex of the maximal cliques.
\begin{figure}[t]
\includegraphics[width=.48\linewidth]{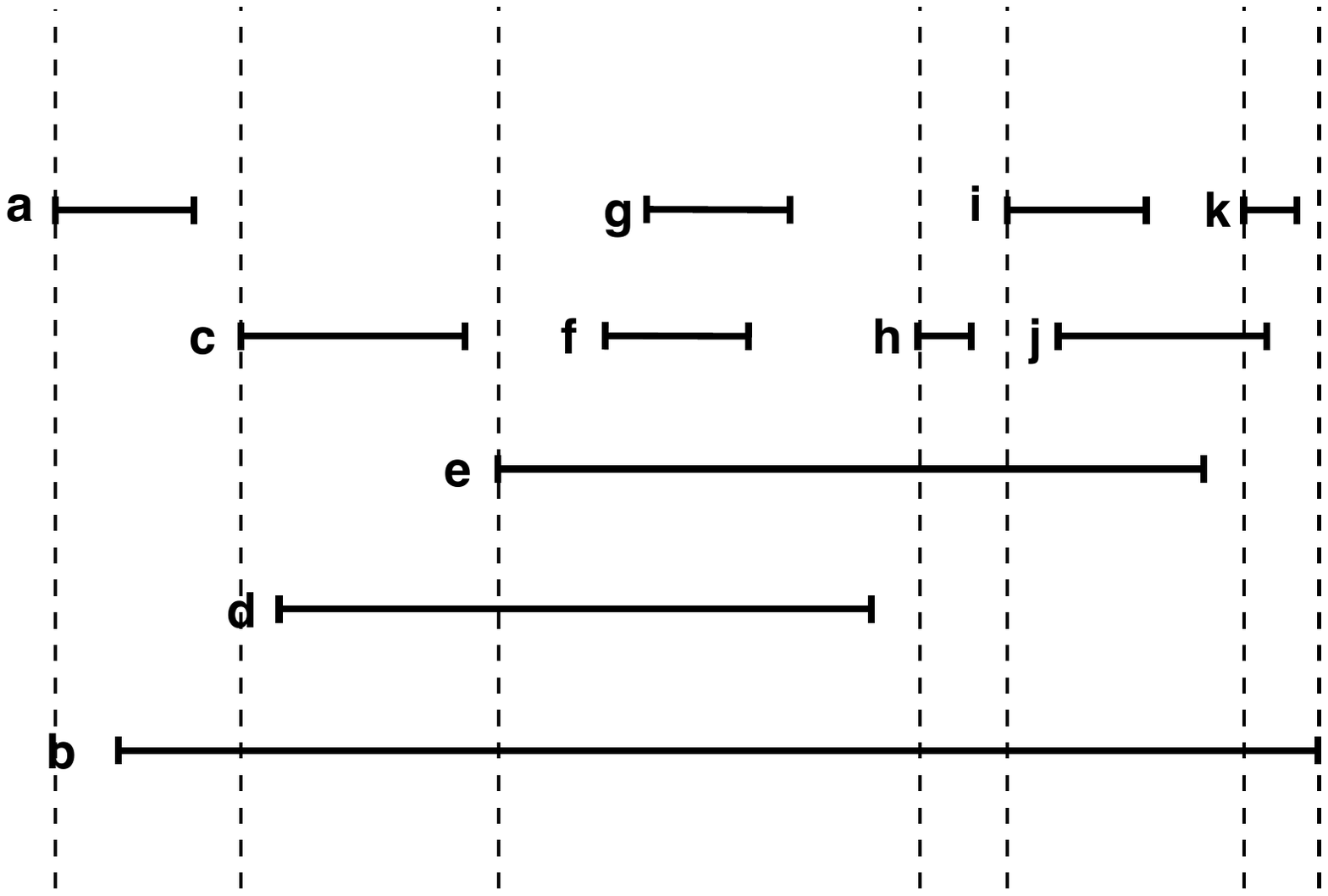}
\includegraphics[width=.48\linewidth]{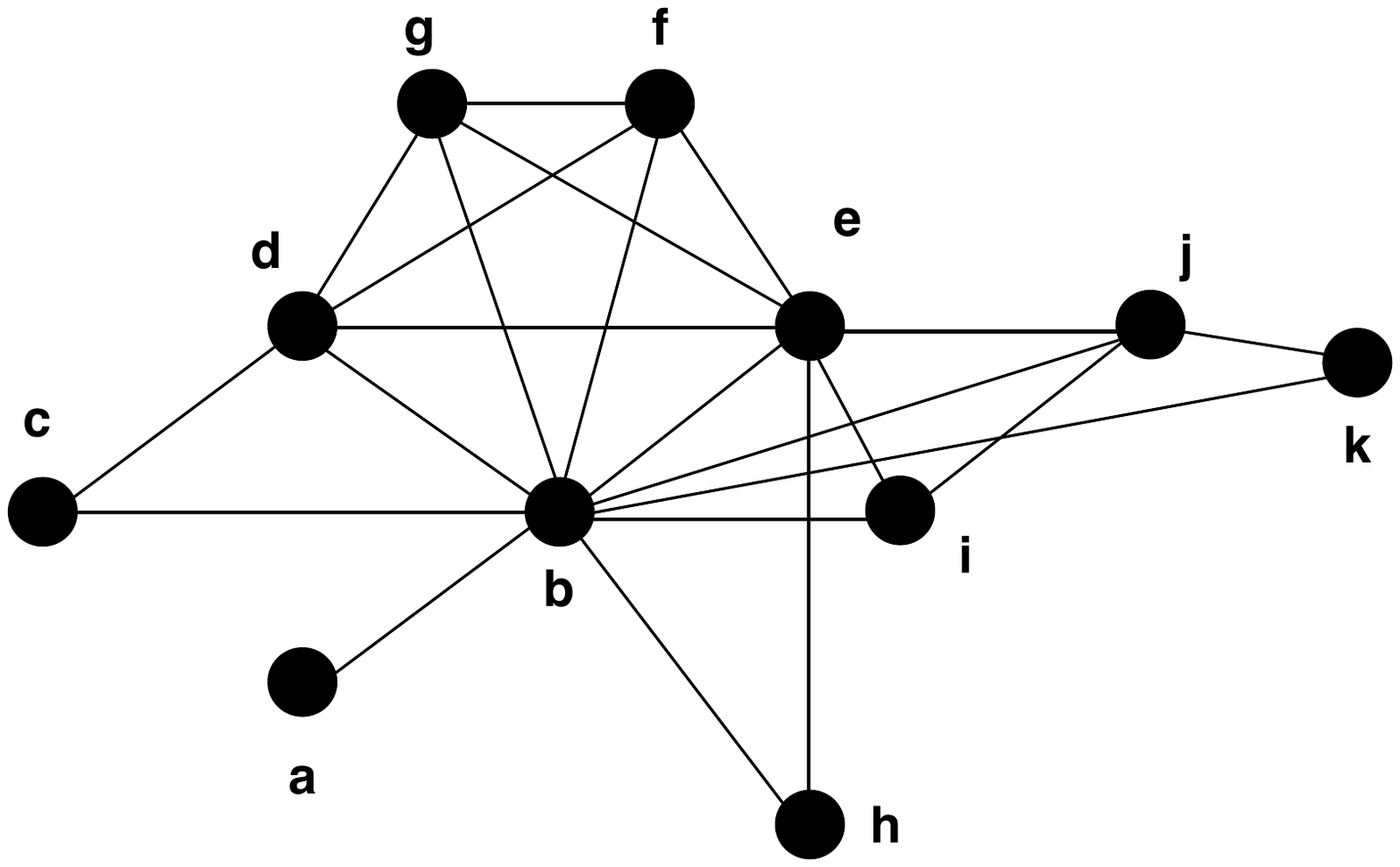}
\begin{center}
\includegraphics[clip=true,trim=0 7.5cm 0 7.5cm,width=.5\linewidth]{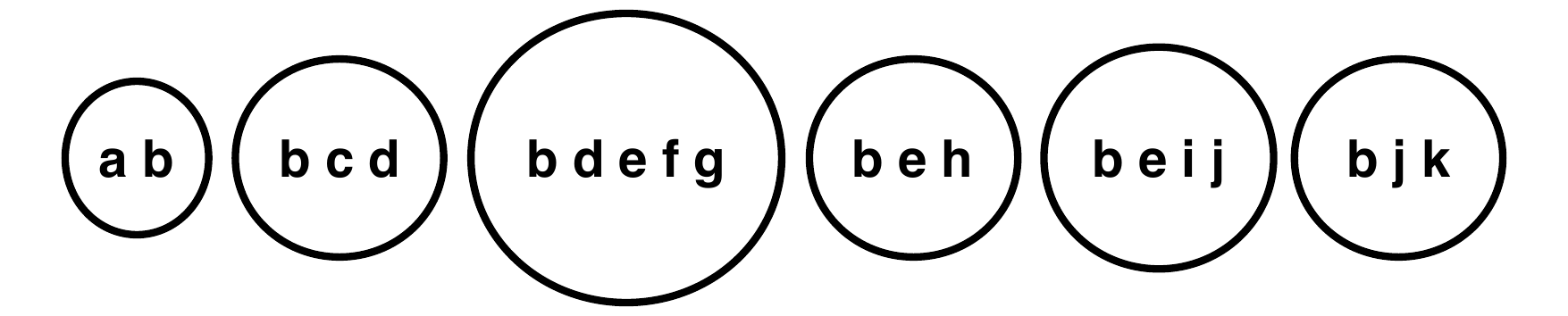}
\end{center}
\caption{The figure on the left is a set of jobs, where the horizontal line segments are the feasible time intervals of jobs and the vertical lines are boundaries of windows. The figure on the right is an interval graph of the corresponding job set. And the figure at the bottom is a set of all the maximal cliques in the interval graph.}
\label{fig:exampleIntervalGraph}
\end{figure}

\subsubsection{Framework of the algorithms}

We propose two exact algorithms, both of which runs in $k$ stages corresponding to each of the $k$ windows.
We maintain a table $\Tl$ that stores all ``valid'' configurations of jobs in all the windows that have been considered so far.
A configuration of a job corresponds to an execution segment.
And a row in the table consists of the configurations of all the jobs.
In addition, for each window~$\win_i$, we compute a table $\Tr_i$ to store all possible configurations of start and end time of jobs available in~$\win_i$.
The configurations in $\Tr_i$ would then be ``concatenated'' to some configurations in $\Tl$ that are ``compatible'' with each other.
These merged configurations will be filtered to remove those non-optimal ones.
The remaining configurations will become the new $\Tl$ for the next window.
To describe the details of the algorithm, we explain several notions below.
We denote by $\win_{\text{left}}$ the union of the windows corresponding to $\Tl$.
More formally, in the $i$-th stage, $\win_{\text{left}} = \cup_{j<i}\win_j$.
And we use $\Tr$ to denote $\Tr_i$ when the context is clear.

\runtitle{Configurations.}
A \emph{configuration}~$\config_i(J)$ of job~$J$ in window~$\win_i$ is an execution segment, denoted by $[\stime_i(J), \etime_i(J))$ contained completely by $\win_i$.
That is, $\stime_i(J)\in \{\bound_i, \bound_i +1, \cdots, \bound_{i+1}-1\} \cup \{\bound_i - 1, \bound_{i+1}\}$ and $\etime_i(J)\in \{\bound_i+1, \bound_i+2, \cdots, \bound_{i+1}\} \cup \{\bound_i, \bound_{i+1} + 1\}$.
Setting $\stime_i(J) = \bound_i - 1$ and $\etime_i(J) = \bound_i$ means $J$ is executed completely before $\win_i$.
Similarly, setting $\stime_i(J) = \bound_{i+1}$ and $\etime_i(J) = \bound_{i+1} + 1$ means $J$ starts execution after $\win_i$.
Also, setting $\stime_i(J) = \bound_i - 1$ and $\etime_i(J) = \bound_{i+1} + 1$ means $J$ starts execution before $\win_i$, crosses the whole window $\win_i$, and ends execution after $\win_i$.
We say that $\stime_i(J) \in \win_i$ or $\etime_i(J) \in \win_i$ if $\stime_i(J) \in [\bound_i, \bound_{i+1})$ or $\etime_i(J) \in (\bound_i, \bound_{i+1}]$ respectively.
And $J$ is executed in $\win_i$ if both $\stime_i(J) \in \win_i$ and $\etime_i(J) \in \win_i$ hold.
For a collection $C$ of jobs, we use $\config_i(C)$ to denote the set of configurations of all jobs in $C$, and $\configL(J)$ and $\configL(C)$ for the counterparts corresponding to $\Tl$. 
The cost of $\config_i(C)$ is the cost corresponding to the execution segments in $\config_i(C)$. That is, $\cost(\config_i(C)) = \sum_{t\in \win_i} (\sum_{J\in C:t\in \config_i(J)}\h(J))^\alpha$.

\runtitle{Validity.}
A configuration~$\config_i(J)$ is \emph{invalid} if one of the following conditions hold:
(i) $\stime_i(J) \geq \etime_i(J)$;
(ii) $\etime_i(J) > \stime_i(J) + \w(J)$ meaning that the length of execution segment of $J$ is larger than the width of $J$;
(iii) $(\etime_i(J) < \stime_i(J) + \w(J)) \land (\stime_i(J) \geq b_i) \land (\etime_i(J) \leq b_{i+1})$ meaning that the length of execution segment of $J$ is smaller than the width of $J$;
(iv) $(\stime_i(J) < \rel(J)) \land (\stime_i(J) < \bound_{i+1})$ meaning that the start time of $J$ is earlier than the release time of $J$;
(v) $(\etime_i(J) > \dl(J)) \land (\etime_i(J) > \bound_i)$ meaning that the end time of $J$ exceeds the deadline of $J$.
Note that for $\configL(J)$, the validity is defined on the boundaries~$\bound_1$ (instead of $\bound_i$) and $\bound_{i+1}$.
And for $\Tl$, $\configL(J)$ is also invalid if $\stime_{\text{left}}(J) = \bound_1 - 1$ since there is no window on the left of $\win_{\text{left}}$.
Similarly, $\config_k(J)$ is invalid if $\etime_k(J) = \bound_{k+1} + 1$.
A configuration $\config_i(C)$ is \emph{invalid} if there exists $J\in C$ such that $\config_i(J)$ is invalid.

\runtitle{Compatibility.}
For job~$J$, the two configurations $\configL(J)$ and $\config_i(J)$ are \emph{compatible} if
(i) $J$ is executed in $\win_{\text{left}}$ for $\configL(J)$, and $J$ is executed before $\win_i$ for $\config_i(J)$;
(ii) $J$ starts execution in $\win_{\text{left}}$ and ends execution after $\win_{\text{left}}$ for $\configL(J)$, and $J$ starts execution before $\win_i$ and ends execution either in $\win_i$ or after $\win_i$ for $\config_i(J)$;
(iii) $J$ is executed completely after $\win_{\text{left}}$ for $\configL(J)$, and $J$ does not start before $\win_i$ for $\config_i(J)$.

\runtitle{Concatenating configurations.}
To concatenate two configurations $\configL(J)$ and $\config_i(J)$, we create a new $\configL(J)$ by the following setting based on the three types of compatible configurations described in the previous paragraph:
for type (i), $\stime_{\text{left}}(J)$ and $\etime_{\text{left}}(J)$ leave unchanged;
for type (ii), $\stime_{\text{left}}(J)$ leaves unchanged and set $\etime_{\text{left}}(J) \gets \etime_i(J)$;
and for type (iii), set $\stime_{\text{left}}(J) \gets \stime_i(J)$ and $\etime_{\text{left}}(J) \gets \etime_i(J)$.
\emph{Concatenating} $\configL(C)$ and $\config_i(C)$ is to concatenate the configurations of each job in $C$. 
The corresponding cost is simply adding the cost of the two configurations. 

\runtitle{Uncertainty and identity.}
A configuration~$\config_i(J)$ is \emph{uncertain} if $\etime_i(J) = \bound_{i+1} + 1$ meaning that the end time of $J$ is not determined yet, and we are not sure at the $i$-th stage whether $\config_i(J)$ will be valid after concatenating $\config_i(J)$ and $\config_{i+1}(J)$.
Two configurations $\config_i(C)$ and $\config'_i(C)$ are \emph{identical} if
(i) $\config_i(J)$ is uncertain if and only if $\config'_i(J)$ is uncertain for all job $J \in C$; and
(ii) the start time of $\config_i(J)$ is equal to the start time of $\config'_i(J)$ for all uncertain configuration $\config_i(J)$ and $J \in C$.
That is, we only consider the differences among the start times of those jobs with uncertain configurations when we distinguish two configurations of a set of jobs.

\subsubsection{An algorithm with three parameters}

\runtitle{Algorithm $\AlgE$.}
The algorithm consists of three components: ListConfigurations, ConcatenateTables and FilterTable.
In the algorithm, we first transform the input job set~$\JS$ to an interval graph, and obtain the maximal cliques~$C_i$ for $1 \leq i \leq k$ and the corresponding windows~$\win_i$.
We start with $\Tl$ containing the only configuration, which sets $\stime_0(J) = \bound_1$ and $\etime_0(J) = \bound_1 + 1$ for all jobs $J$.
That is, the configuration treats all the jobs to be not yet executed.
Then we visit the windows from the left to the right.

\noindent
\textbf{\emph{ListConfigurations:}}
For window~$\win_i$ and jobs in $C_i$, we construct $\Tr$ storing all configurations of $J \in C_i$.
We enumerate all $\stime_i(J) \in \win_i$ and $\etime_i(J) \in \win_i$ for each job $J \in C_i$, list all the combinations of all the jobs~$J$ with all of its start times and end times, and store the results in $\Tr$ in the way that one row is for one configuration $\config_i(C_i)$.
In another words, $\Tr$ stores all the combinations of execution segments in $\win_i$ for all jobs $J \in C_i$.
Note that the jobs with release time later than $\win_i$ are considered to execute after $\win_i$ and the jobs with deadline earlier than $\win_i$ are considered to execute before $\win_i$.
For each configuration~$\config_i(C_i)$, we also store its cost contribution $\cost(\config_i(C_i))$ together.
We also check each of the configurations and delete those invalid ones.

\noindent
\textbf{\emph{ConcatenateTables:}}
We then concatenate compatible configurations in $\Tl$ and $\Tr$. 
The resulting table is the new $\Tl$.
More specifically, for each configuration~$\configL(C)$ in $\Tl$ and each configuration~$\configR(C)$ in $\Tr$, we concatenate $\configL(C)$ and $\configR(C)$ if they are compatible, and store the result to a new row in $\Tl$.
We also check each of the configurations in the new $\Tl$ and delete those invalid ones.

\noindent
\textbf{\emph{FilterTable:}}
After concatenation, we filter non-optimal configurations.
We classify all the configurations in $\Tl$ into groups such that the configurations in a group are identical and no two configurations from different groups are identical.
For each group, we only leave the configuration with the lowest cost (choosing anyone to break tie if any) and remove the others in the group.
In the current $\Tl$, no two configurations are identical.

After processing all windows, the only configuration in the final $\Tl$ is returned as the solution. Algorithm~\ref{alg:fpt} is the pseudocode of this algorithm.

\begin{algorithm}
\caption{The fixed parameter algorithm $\AlgE$}
\label{alg:fpt}
\begin{algorithmic}
\State \textbf{Input:} a set of job $\JS$
\State \textbf{Output:} an optimal configuration of $\JS$
\State $\{(\win_i, C_i)\}_{i=1}^k \gets$ the windows and their corresponding cliques of $\JS$
\State $\Tl \gets$ a configuration that sets all jobs $J \in \JS$ to be not yet executed
\For{$i$ from $1$ to $k$}
	\State $\Tr \gets$ ListConfigurations($\win_i$, $C_i$)
	\State $\Tl \gets$ ConcatenateTables($\Tl$, $\Tr$)
	\State $\Tl \gets$ FilterTable($\Tl$)
\EndFor
\Return any configuration in $\Tl$
\end{algorithmic}
\end{algorithm}

\begin{lemma}
Algorithm~$\AlgE$ outputs an optimal solution.
\label{lm:alg_e_opt}
\end{lemma}

\begin{proof}
In each stage, we list all possible configurations.
A configuration is deleted only when it is invalid or it is identical to another configuration with lower cost.
Hence, an invalid configuration cannot be optimal.
So we focus on the other case.
Given two identical configurations $\config_{\text{left}}(C)$ and $\config'_{\text{left}}(C)$ with $\cost(\config_{\text{left}}(C)) < \cost(\config'_{\text{left}}(C))$, we show that $\config'_{\text{left}}(C)$ cannot be optimal.
Suppose there is an optimal solution $\config^*$ containing $\config'_{\text{left}}(C)$, which means each execution segment $\config'_{\text{left}}(J)$ in $\config'_{\text{left}}(C)$ is completely contained by the corresponding execution interval of $J$ in $\config^*$.
Since $\config_{\text{left}}(C)$ and $\config'_{\text{left}}(C)$ are identical, the start times of $J$ are the same in the two configurations for all uncertain jobs $J$.
In $\win_{\text{left}}$, this means the uncertain jobs do not make the costs of the two configurations to be different, and the jobs $\JS_c$ that are not uncertain do.
Note that $\JS_c$ is consisted of the jobs with their end times being determined.
This means we can replace the configurations of $\JS_c$ in $\config'_{\text{left}}(C)$ by the configurations of $\JS_c$ in $\config_{\text{left}}(C)$ and this action will not affect the procedures in the algorithm thereafter.
However, this also results in a solution of lower cost and contradicts the assumption that $\config^*$ is optimal.
Thus $\config'_{\text{left}}(C)$ cannot be optimal.
Therefore, none of the deleted configuration can be part of an optimal schedule.
That is, no optimal schedule would be removed through out the whole process.
\end{proof}

\begin{theorem}
Algorithm~$\AlgE$ computes an optimal solution in
$O(k \cdot \wmax^{2m} \cdot (\win_{\max} + 1)^{4m} \cdot n)$ time,
where $n$ is the number of jobs,
$\wmax$ is the maximum width of jobs,
$m$ is the maximum size of cliques,
$\win_{\max}$ is the maximum length of windows,
and $k$ is the number of windows.
\label{thm:fpt1}
\end{theorem}

\begin{proof}
We first compute the time complexities for the three components of the algorithm, and then compute the total time complexity.
For ListConfigurations, there are at most $(\win_{\max} + 1)^{2m}$ configurations in the outputted table $\Tr$, since there are at most $\win_{\max} + 1$ possible start times and end times respectively and at most $m$ jobs that should be considered in the current window.
For each configuration, it takes $O(n)$ time for construction and validity checking.
It also takes $O(n \win_{\max})$ to compute the cost of a configuration.
So, the time complexity for ListConfigurations is
\[
O((\win_{\max} + 1)^{2m} \cdot n \win_{\max})
=
O((\win_{\max} + 1)^{2m+1} \cdot n)
\enspace .
\]

Before computing the time complexities of the other components, we focus on the number of configurations of $\Tl$ at the end of each iteration in the algorithm.
Since $\Tl$ is filtered to have no identical configurations, the number of configurations can be upper bounded.
This number depends on the number of different start times of uncertain jobs.
There are at most $m$ uncertain jobs, and for each such job, the number of start times is at most $\wmax$.
Note that the end times of these jobs are all set to be later than the current window and will not affect the number of configurations.
So the number of configurations of $\Tl$ at the end of each iteration is at most $\wmax^m$.

For ConcatenateTables, there are at most $\wmax^m \cdot (\win_{\max} + 1)^{2m}$ configurations in the outputted table $\Tl$.
This is because for each configuration in the input $\Tl$, we need to compare it with all the configurations in $\Tr$ for compatibility checking.
For each configuration, it takes $O(n)$ time for compatibility checking, concatenation and validity checking.
Thus the time complexity for ConcatenateTables is $O(\wmax^m \cdot (\win_{\max} + 1)^{2m} \cdot n)$.

For FilterTable, the number of configurations in the outputted table~$\Tl$ is at most the number of configurations outputted by ConcatenateTables.
Also, the number of groups is at most its number of configurations.
Thus it takes
\[
O([\wmax^m \cdot (\win_{\max} + 1)^{2m}]^2 \cdot n)
=
O(\wmax^{2m} \cdot (\win_{\max} + 1)^{4m} \cdot n)
\]
time for classification.
And it takes $O(\wmax^m \cdot (\win_{\max} + 1)^{2m})$ time for deletion.
So the time complexity for FilterTable is $O(\wmax^{2m} \cdot (\win_{\max} + 1)^{4m} \cdot n)$.
Since there are $k$ iterations, the total time complexity is $O(k \cdot \wmax^{2m} \cdot (\win_{\max} + 1)^{4m} \cdot n)$.
\end{proof}

In the worst case, there are at most $O(n)$ windows. So algorithm~$\AlgE$ also runs in $f(\wmax, m, \win_{\max}) \cdot O(n^2)$ time where $f(\wmax, m, \win_{\max}) = \wmax^{2m} \cdot (\win_{\max} + 1)^{4m}$.

\begin{corollary}
Grid problem is fixed parameter tractable with respect to
the maximum width of jobs,
the maximum number of overlapped feasible intervals,
and the maximum length of windows.
\end{corollary}

\subsubsection{An algorithm with two parameters}
\label{sec:exact_fpt2}

This section describes how to drop out the parameter $\win_{\max}$ in the previous algorithm by generalizing the definitions of windows and boundaries.

At the beginning of Algorithm~$\AlgE$, we transform a set of jobs to its corresponding interval graph and obtain a sequence of windows by the set of maximal cliques in the interval graph.
We require in the algorithm that all the cliques should be maximal.
However, the algorithm is still optimal and has parameterized bound of time complexity if we divide a maximal clique into multiple non-maximal cliques in a specific way.
Given a maximal clique $C_i$ and its corresponding window $\win_i$, we divide $\win_i$ into a set of contiguous windows $\win_{i_1}, \win_{i_2}, \ldots$ such that $\win_i = \cup_j \win_{i_j}$.
Note that the set of jobs $C_{i_j}$ corresponding to $W_{i_j}$ is a clique in the interval graph since $C_i$ is a clique and $C_{i_j} \subseteq C_i$.
In this way, the number of jobs in the window $W_{i_j}$ is still at most $m$.
Furthermore, since this window division does not affect the proof of lemma~\ref{lm:alg_e_opt}, the algorithm is still optimal.
Thus we have the following observation.

\begin{observation}
Algorithm~$\AlgE$ outputs an optimal solution if it receives a set of contiguous windows containing all the jobs such that each window represents a clique (not necessarily maximal) in the interval graph of the input jobs.
And we have the number of jobs in each window is at most the maximum number of overlapped feasible intervals.
\label{obs:alg_e_general}
\end{observation}

To drop out the parameter $\win_{\max}$ in the previous algorithm, we divide windows into smaller ones such that the number of configurations in a window can be bounded by $\wmax$ and $m$.
In the new algorithm, we set the locations of boundaries at the release times and deadlines of all the jobs and construct the windows bases on these boundaries.
In this way, there is no job being released or attaining its deadline in the middle of a window, and all the jobs in the window can be put anywhere in the window.
Thus the number of used timeslots is at most $m \cdot \wmax + 2 (\wmax - 1)$.
This is because in the worst case, all jobs in a window are scheduled such that no job overlaps to another and these jobs consume at most $m \cdot \wmax$ timeslots.
In addition, we need to consider the cases that a job's start time is earlier than the window or its deadline is later than the window.
Both cases consume at most $\wmax - 1$ timeslots respectively.
Note that this window division results in a set of windows that their sizes are smaller than their original counterparts, and thus observation~\ref{obs:alg_e_general} can be applied.
Based on this new window division, we have the following algorithm.

\runtitle{Algorithm~$\AlgEPlus$.}
This algorithm is similar to algorithm~$\AlgE$ except the definitions of boundaries and the component ListConfigurations.
Given a set of jobs $\JS$, the algorithm uses the set of boundaries $\{\rel(J) \mid J \in \JS\} \cup \{\dl(J) \mid J \in \JS\}$ to construct the windows and obtain the corresponding cliques.
Let $k$ denotes by the number of windows.
There are $k$ stages for the algorithm.
At the $i$-th stage, the algorithm runs ListConfigurations, ConcatenateTables and FilterTable accordingly as algorithm~\ref{alg:fpt} does.
It finally outputs the only configuration in $\Tl$.
For the component ListConfigurations, we only consider to schedule jobs on the timeslots used instead of enumerating all possibilities of start times and end times.
The algorithm tries all $m \cdot \wmax$ timeslots (the worst case described in the previous paragraph) as the start time of a job, and also the $2 (\wmax - 1)$ schedules that a job is partially executed in the window.
In addition, the component shall includes the cases that either a job is completely executed before the window, it is completely executed after the window, or it crosses the window.

\begin{theorem}
Algorithm~$\AlgEPlus$ computes an optimal solution in
$f(\wmax, m) \cdot O(n^2)$ time,
where $n$ is the number of jobs,
$\wmax$ is the maximum width of jobs,
$m$ is the maximum size of cliques,
and $f(\wmax, m) = (4 m \cdot \wmax^2)^{2m}$.
\end{theorem}

\begin{proof}
As in the proof of theorem~\ref{thm:fpt1}, we compute the running time of the three components and then the total time complexity.
For the component ListConfigurations, there are at most $(m \cdot \wmax + 2 (\wmax - 1) + 3)^m$ outputted configurations, since there are at most $m \cdot \wmax + 2 (\wmax - 1) + 3$ schedules for a job (see the description in the previous paragraph) and at most $m$ jobs in a window.
It takes $O(n (m \cdot \wmax + 2 (\wmax - 1))) \leq O(n \cdot m \cdot \wmax)$ time to compute the cost for each configuration.
Thus the time complexity for ListConfigurations is at most
\[
O((m \cdot \wmax + 2 (\wmax - 1) + 3)^m \cdot (n \cdot m \cdot \wmax))
\leq
O((4 m \cdot \wmax)^{m+1} \cdot n)
\enspace .
\]

The time complexities of ConcatenateTables and FilterTable are similar to that in the proof of theorem~\ref{thm:fpt1} except the number of outputted configurations.
For ConcatenateTables and FilterTable, both the number of outputted configurations are at most $\wmax^m \cdot (4 m \cdot \wmax)^m$.
Thus their running time are at most $O(\wmax^{2m} \cdot (4 m \cdot \wmax)^{2m} \cdot n)$.
Since there are $k = O(n)$ iterations, the total time complexity of the algorithm is at most
\[
O((4 m \cdot \wmax^2)^{2m} \cdot n^2)
=
f(\wmax, m) \cdot O(n^2)
\enspace .
\qedhere
\]
\end{proof}

\begin{corollary}
Grid problem is fixed parameter tractable with respect to
the maximum width of jobs,
and the maximum number of overlapped feasible intervals.
\end{corollary}

%% file: exact_exp.tex
\subsection{An exact algorithm without parameter}
\label{sec:exact_exp}

For the case with unit width and arbitrary height of Grid problem, we can use algorithm~$\AlgE$ to design an exact algorithm that its time complexity is only measured in the size of the input.

In Algorithm~$\AlgE$, we maintain two tables $\Tl$ and $\Tr$ for each stage.
At each stage, the core operations are to construct $\Tr$, merge $\Tl$ and $\Tr$, and filter the resulting table.
In the case with unit width and arbitrary height, one may observe that the functionalities of these core operations are not affected by the length of the windows representing $\Tl$ and $\Tr$.
For example, we can restrict the window length to be a constant but not be related to the cliques in the interval graph, and the algorithm still works correctly.
By fixing the lengths of all windows, a new exact algorithm is obtained.
Without loss of generality, we assume that the number of timeslots~$\tau$ is even.
We enforce all windows to have length 2, i.e. we have $\tau / 2$ windows in total.
By this setting, the new algorithm runs in $O((\tau / 2) \cdot 4^{2n} \cdot n)$ time where $n$ is the number of jobs.
This is because the numbers of configurations for the three components in the algorithm are at most $4^n$.
Note that the input size $N$ of the problem is $3 n \log \tau + n \log \hmax$ where $\hmax$ is the maximum height over all jobs.
Since $\log \tau = O(N)$, the running time becomes $2^{O(N)}$.
Thus we have the following theorem.

\begin{theorem}
There is an exact algorithm running in $2^{O(N)}$ time for the Grid problem with unit width and arbitrary height where $N$ is the length of the input.
\end{theorem}

Our algorithm is highly more efficient than a brute force search.
Such naive method would enumerate all possible schedules and check if they are feasible and optimal, which requires $O(\tau^n n)$ time.
The running time can be rewritten as $2^{O(Nn)}$ or more clearly, $(2^{O(N)})^n$.
The exact algorithm modified from our fixed parameter algorithm indeed crosses out an `$n$' in the exponent.

%% file: exact_lb.tex
\subsection{Lower bound}
\label{sec:exact_lb}

This section provides two lower bounds on the running time of the Grid problem under a certain condition. 

Jansen et al. \cite{DBLP:journals/siamdm/JansenLL16} derived several lower bounds for scheduling and packing problems which can be used to develop lower bounds for our problem.
Their lower bounds assume \emph{Exponential Time Hypothesis} (ETH) holds, which conjectures that there is a positive real $\epsilon$ such that 3-\textsc{Sat} cannot be decided in time $2^{\epsilon n} N^{O(1)}$ where $n$ is the number of variables in the formula and $N$ is the length of the input.
A lower bound for other problems can be shown by making use of strong reductions, i.e. reductions that increase the parameter at most linearly.
Through a sequence of strong reductions, they obtain two lower bounds for \textsc{Partition}, $2^{o(n)} N^{O(1)}$ and $2^{o(\sqrt{N})}$ where $n$ is the cardinality of the given set and $N$ is the length of the input.

\runtitle{Reduction.}
We design a strong reduction from \textsc{Partition} to the decision version of Grid problem with unit width and arbitrary height.
Here is a sketch of the reduction.
Recall that \textsc{Partition} is a decision problem that decides if a given set $S$ of integers can be partitioned into two disjoint subsets such that the two subsets have equal sum.
For each integer $s \in S$, we convert it to a job $J$ with $\rel(J) = 0$, $\dl(J) = 2$, $\w(J) = 1$ and $\h(J) = 2s$.
We claim that $S$ is a partition if and only if the set of jobs $J$ can be scheduled with cost at most $2 (\sum_{s \in S} s)^\alpha$.
Note that the specified cost appears when jobs can be put into two timeslots with equal loads.
By setting the length of the input as the parameter, we observe that the parameter increases at most linearly from \textsc{Partition} to our problem. 
(Note that a strong reduction from \textsc{Partition} to the case with unit height and arbitrary width can be done similarly, and the results also apply on that case.)
Furthermore, we can choose the number of jobs as a parameter of the problem.
Note that the reduction above does not increase this parameter with respect to the parameter of \textsc{Partition}, which is the number of integers. 

\begin{theorem}
There is a lower bound of $2^{o(\sqrt{N})}$ and a lower bound of $2^{o(n)} N^{O(1)}$
on the running time for the Grid problem unless ETH fails, where $n$ is the number of jobs and $N$ is the length of the input.
\end{theorem}

%% file: peak.tex
\renewcommand{\Alg}{{\mathcal{G}}}

\section{Minimizing peak and non-preemptive machine minimization}
\label{sec:peak}

In this section, we investigate extension of our solutions to other objectives and other problems.
In particular, we consider the objective of minimizing peak electricity cost in the $\grid$ model
and we name it the $\peak$ problem.
We also consider the classical non-preemptive machine minimization problem denoted as $\minimization$.

\runtitle{The $\peak$ problem.} 
The input is the same as the $\grid$ problem.
The goal is to find a feasible non-preemptive schedule such that the maximum load over the time horizon is minimized. 
$\peak$ has been proven to be NP-hard~\cite{tang2013smoothing}
and approximation algorithms are known for requests having common feasible interval with approximation ratio~$4$
and for requests having agreeable deadlines with approximation ratio $O(\log \frac{\wmax}{\wmin})$.

\runtitle{The $\minimization$ problem.} 
The input is a set of jobs each with a processing time $\Work{}$, release time $\Rtime{}$, and deadline $\Dline{}$. 
Each job has to be scheduled non-preemptively on one of the (infinite number of) machines.
For each machine, at most one job can be executed at any time. 
The goal is to minimize the number of machines used.

The $\minimization$ problem can be seen as a special case of the $\peak$ problem where jobs have unit height.
A lower bound of $\log_3 \frac{\wmax}{{\wmin}}$ on the competitive ratio of any online algorithm has been shown in~\cite{DBLP:conf/fsttcs/Saha13}.
As a result, this lower bound also applies to $\peak$.
The author also provided an $O(\log \frac{p_{\max}}{p_{\min}})$-competitive algorithm for the $\minimization$ problem. 
The algorithm classifies jobs by processing times and applies the constant approximate algorithm
for each class.\footnote{The paper did not state explicitly which algorithm to use and a result~\cite{DBLP:conf/approx/ChuzhoyC09} cited in the paper 
has later been retreated by the same authors~\cite{DBLP:conf/approx/ChuzhoyC09a}.}
The employed algorithm can be the $6$-competitive algorithm proposed by Yu and Zhang~\cite{DBLP:journals/orl/YuZ09} or our algorithm for uniform widths jobs (to be analyzed in Section~\ref{subsubsec:uniform_width}). 

In this section, we show that our online algorithm solves the $\peak$ problem optimally in an asymptotical sense and provide an alternative asymptotically optimal competitive algorithm for the $\minimization$ problem.


\input{peak_online}

\input{peak_exact}

%% file: peak_online.tex
\subsection{Online algorithms}
\label{sec:peak_online}

In this section, 
we prove that the online algorithm~$\Alg$ proposed in Section~\ref{sec:general} is asymptotically optimal for the $\peak$ problem.
We first state two properties, one for $\newBKP$ that~$\Alg$ is based on and the other for the optimal algorithm w.r.t.\ the peak objective.
Let function $\p(\Sch)$ denote the maximum load (cf.\ speed) of any schedule $\Sch$,
i.e., $\p(\Sch) = \max_{t} \Load{\Sch}{t}$. 
Combining Lemma~\ref{lm:bkp_load} and the fact that $\BKP$ is $e$-competitive w.r.t.\ maximum speed~\cite{DBLP:journals/jacm/BansalKP07},
we have the following observation.

\begin{observation}
\label{lem:BKPL_peak}
The $\newBKP$ algorithm is $e(1+e)$-competitive with respect to maximum speed.
\end{observation}

\begin{proof}
By Lemma~\ref{lm:bkp_load}, for any integral $t$ and $0<\Delta<1$, $\load(\BKP, t+\Delta) \leq (1+e)\cdot \load(\BKP,t)$. 
Also, by $\newBKP$, $\load(\newBKP,t) \geq \load(\BKP,t+\Delta)$ for any $0<\Delta<1$. 
Hence, $\p(\newBKP) \leq (1+e)\cdot \p(\BKP)$.
\end{proof}

On the other hand, $\YDS$ guarantees that the maximum speed is minimized~\cite{DBLP:journals/jacm/BansalKP07}.
Similar to Observation~\ref{thm:opt_dvs_grid}, $\YDS$ gives a lower bound for the $\peak$ problem. 

\begin{observation}
\label{ob:opt_peak}
Let $\Opt_D$ and $\Opt_G$ be the optimal schedule for the $\dvs$ and $\peak$ problem, respectively. 
Given a job set $\Jobsetij{G}$ for the $\peak$ problem. 
Let $\Jobsetij{D}$ denote the job set after converting $\Jobsetij{G}$ into a job set for the $\dvs$ problem. 
Then, $\p(\Opt_D(\Jobsetij{D})) \leq \p(\Opt_G(\Jobsetij{G}))$. 
\end{observation}

Recall that in Sections~\ref{sec:general_unit_width} and~\ref{sec:general},
we have presented three algorithms 
$\AlgV$, $\AlgUV$, and $\Alg$ for unit-width jobs, uniform-width jobs, and arbitrary width jobs, respectively.
Here, we analyze their performance w.r.t.\ $\peak$.
In summary, we show that for $\peak$, we have

\begin{itemize}
\item $\AlgV$ is $2(e+e^2)$-competitive for unit-width job sets (Theorem~\ref{thm:algV_peak});
\item $\AlgUV$ is $(6(e+e^2)+3)$-competitive for unit-width job sets (Theorem~\ref{thm:UV_peak}); and
\item $\Alg$ is $(18(e+e^2)+9)\cdot \lceil\log \frac{\wmax}{\wmin} \rceil$-competitive for arbitrary-width job sets  (Theorem~\ref{thm:alg_peak}).
\end{itemize}

\subsubsection{Unit-width jobs.}
Recall that for each timeslot t, $\AlgV$ schedules jobs to start at $t$ 
such that $\Load{\AlgV}{t}$ is at least $\Load{\newBKP}{t} = (1+e)\cdot\Load{\BKP}{t}$ or until all available jobs have been scheduled.
We prove that although $\Load{\AlgV}{t}$ might be higher than $\Load{\newBKP}{t}$, 
the peak of $\AlgV$ is no more than $2(e+e^2)$ times of the peak of the optimal. 

\begin{theorem}
\label{thm:algV_peak}
For any job set $\Jobsetij{}$ where for each job has unit width, 
$\p(\AlgV(\Jobsetij{})) \leq 2(e+e^2)\cdot \p(\Opt(\Jobsetij{}))$. 
\end{theorem}

\begin{proof}
Let $\hmax(\AlgV,t)$ be the maximum height of jobs scheduled at $t$ by $\AlgV$; we set $\hmax(\AlgV, t)= 0$ if $\AlgV$ assigns no job at $t$.
We classify each timeslot $t$ into two types: 
(i) $\hmax(\AlgV,t) < \Load{\newBKP}{t}$, and (ii) $\hmax(\AlgV,t) \geq \Load{\newBKP}{t}$. 
We denote by $\I_1$ and $\I_2$ the union of all timeslots of Type (i) and (ii), respectively. 
(Notice that $\I_1$ and $\I_2$ can be empty and the union of $\I_1$ and $\I_2$ covers the entire time line.) 

We first prove that for any job set $\Jobsetij{}$ where for each job $\Job{} \in \Jobsetij{}$, $\Width{} =1$, 
$\p(\AlgV(\Jobsetij{}),\I_1) \leq 2(e+e^2)\cdot \p(\Opt(\Jobsetij{}))$; and $\p(\AlgV(\Jobsetij{}),\I_2) \leq 2 \cdot \p(\Opt(\Jobsetij{}))$. 
For every timeslot $t \in \I_1$, $\Load{\AlgV}{t} < \Load{\newBKP}{t} + \hmax(\AlgV,t) \leq 2\cdot \Load{\newBKP}{t}\leq2(1+e)\cdot\Load{\BKP}{t}$. 
By definition, let $t$ bet the timeslot where $\Load{\AlgV}{t} \geq \Load{\AlgV}{t'}$ for any other $t'$, $\p(\AlgV, \I_1) = \Load{\AlgV}{t} \leq 2(1+e)\cdot \Load{\BKP}{t}$. 
Therefore, $\p(\AlgV, \I_1)\leq 2e(1+e)\cdot\p(\Opt)$ by Observation~\ref{ob:opt_peak} and Lemma~\ref{lem:BKPL_peak}.
For every timeslot $t \in \I_2$, $\Load{\AlgV}{t} < \Load{\newBKP}{t} + \hmax(\AlgV,t) \leq 2\cdot \hmax(\AlgV,t)$. 
In the optimal schedule, the job with height $\hmax(\AlgV,t)$ has to be scheduled somewhere or the schedule is not feasible, so $\p(\Opt) \geq \hmax(\AlgV,t)$. 
Hence, $\p(\AlgV, \I_2) = \Load{\AlgV}{t} \leq 2\hmax(\AlgV,t) \leq 2\cdot \p(\Opt)$ where $t$ is the timeslot with peak power.

Since $\I_1$ and $\I_2$ are disjoiont, $\p(\AlgV) =\max\{\p(\AlgV, \I_1), \p(\AlgV, \I_2)\}$. 
Therefore, $\p(\AlgV) = \max\{2(e+e^2)\cdot \p(\Opt), 2\cdot \p(\Opt)\} = 2(e+e^2)\cdot \p(\Opt)$.
\end{proof}

\subsubsection{Uniform-width jobs.}
\label{subsubsec:uniform_width}
Recall that in handling uniform-width jobs, we classify jobs into tight and loose jobs.
Let $\niceJobsetij{}$ denote the input set with uniform width jobs, 
$\niceJStight$ and $\niceJSloose$ denote the set of tight jobs and loose jobs in $\niceJobsetij{}$, respectively. 
We first prove that any feasible schedule for tight jobs is $3$-competitive due to the ``inflexibility'' (Lemma~\ref{lm:tight_peak}).
Then, we prove that $\AlgUV$ is $O(1)$-competitive for loose jobs (Lemma~\ref{lm:loose_peak}).

\begin{lemma}
\label{lm:tight_peak}
For any feasible schedule $\Schij{}$, $\p(\Schij{}(\niceJStight)) \leq 3\cdot \p(\Opt(\niceJobsetij{}))$.
\end{lemma}

\begin{proof}
We prove it by showing that even if the execution intervals of jobs are considered as the whole feasible interval, the peak is not too much larger than the peak in the optimal schedule. 

We first \emph{extend} jobs $\Job{}\in\niceJStight$ to $J^+$ as follows: $\niceRtime{} = \Rtime{}$, $\niceDline{}= \Dline{}$, $\niceWidth{}=\Dline{}-\Rtime{}$, and $\niceHeight{} = \Height{}$. 
That is, every job has its width as the length of its feasible interval. 
We denote the resulting job set by $\Jobset^+$. 
Since each job in $\Jobset^+$ are not shiftable, there is only one feasible schedule for $\Jobset^+$ and it is optimal. 
It is clear that $\p(\Sch(\niceJStight)) \leq \p(\Opt(\Jobset^+))$. 

Similar to Lemma~\ref{thm:cost_UV1} (i), we bound the load at time $t$ of $\Opt(\Jobset^+)$ by the loads of constant number of timeslots in $\Sch(\niceJStight)$. 
Consider the job $\Job{}$ corresponding to $J^+$, the execution interval of~$\Job{}$ in any feasible schedule must contain either timeslot $t-(\w-1)$, $t+(\w-1)$, or $t$.  
Hence, we can upper bound the load at time $t$ in $\Opt(\Jobset^+)$ as follows: 
$\Load{\Opt(\Jobset^+)}{t}\leq \Load{\Opt(\niceJStight)}{t-(\w-1)}+\Load{\Opt(\niceJStight)}{t+(\w-1)} + \Load{\Opt(\niceJStight)}{t}$.
Hence, $\p(\Sch(\niceJStight)) \leq\p(\Opt(\Jobset^+)) \leq 3\cdot \p(\Opt(\niceJSloose))$.
\end{proof}

\comment{
\hhl{For the loose jobs set $\niceJSloose$ where jobs have uniform width $\w$, we transform it to $\alignJobsetij{}$ by $\convertfi$ into jobs with release times and deadlines being at $i\cdot \w$ for certain integers $i$. 
The schedule $\AlgUV(\niceJSloose)$ is transformed from the schedule $\AlgV(\alignJobsetij{})$ by $\freeS$. 
We prove that $\AlgUV(\niceJSloose)$ is constant competitive (Lemma~\ref{lm:loose_peak}). 
For the tight jobs in $\niceJStight{}$, each job is executed once it is released. We prove that any feasible schedule for tight jobs is $3$-competitive due to the ``inflexibility'' (Lemma~\ref{lm:tight_peak}).}
}

\begin{lemma}
\label{lm:loose_peak}
For loose jobs set $\niceJSloose$ where jobs have uniform width, $\p(\AlgUV(\niceJSloose)) \leq 6(e+e^2) \cdot \p(\Opt(\niceJobsetij{}))$. 
\end{lemma}

\begin{proof}

According to $\AlgUV$, the job set $\niceJSloose$ is transformed into a job set $\alignJobsetij{}$ by $\convertfi$ and $\AlgV$ is run on~$\alignJobsetij{}$. 
Then, the schedule $\AlgV(\alignJobsetij{})$ is transformed to a schedule of $\niceJSloose$ by Transformation $\freeS$. 
Hence, by Theorem~\ref{thm:algV_peak}, $\p(\AlgUV(\niceJSloose)) \leq \p(\AlgV(\alignJobsetij{}))\leq 2(e+e^2)\cdot \p(\Opt(\alignJobsetij{}))$.

By Observation~\ref{thm:align_load}, given the optimal schedule of $\niceJSloose$, we have shown the load at any time in the schedule $\alignSchij{}$ generated by $\alignS$ is no more than the sum of the load of the optimal schedule at three timeslots and hence no more than three times the peak of the optimal schedule. 
Therefore, $\p(\Opt(\alignJobsetij{})) \leq \p(\alignSchij{}) \leq 3 \cdot \p(\Opt(\niceJSloose))$

In summary, $\p(\AlgUV(\niceJSloose))\leq 2(e+e^2)\cdot \p(\Opt(\alignJobsetij{}))\leq 6(e+e^2) \cdot \p(\Opt(\niceJSloose))$.
%
\end{proof}


\begin{theorem}
\label{thm:UV_peak}
For any jobs set $\niceJobsetij{}$ where jobs have uniform width, $\p(\AlgUV(\niceJobsetij{})) \leq (6(e+e^2)+3)\cdot \p(\Opt(\niceJobsetij{}))$. 
\end{theorem}

\begin{proof}
By definition, $\p(\AlgUV(\niceJobsetij{})) \leq \p(\AlgUV(\niceJStight)) + \p(\AlgUV(\niceJSloose))$. 
By Lemma~\ref{lm:tight_peak} and~\ref{lm:loose_peak}, $\p(\AlgUV(\niceJobsetij{})) \leq 3\cdot\p(\Opt(\niceJobsetij{})) + 6(e+e^2)\cdot \p(\Opt(\niceJobsetij{})) = {(6(e+e^2)+3)}\cdot \p(\Opt(\niceJobsetij{}))$. 
\end{proof}

\subsubsection{Arbitrary width jobs.}
Finally, we bound the competitive ratio of $\Alg$.

\begin{theorem}
\label{thm:alg_peak}
For any job set $\Jobsetij{}$, {$\p(\Alg(\Jobsetij{})) \leq {(18(e+e^2)+9)}\cdot \lceil\log \frac{\wmax}{\wmin} \rceil \cdot \p(\Opt(\Jobsetij{}))$}.
\end{theorem}

\begin{proof}
Recall that $\Alg(\Jobset)$ partition jobs into subsets $\Jobsetij{p}$ such that in each $\Jobsetij{p}$ jobs have bounded widths. 
For each $\Jobsetij{p}$, it is transformed into $\niceJobsetij{p}$ and $\AlgUV$ is applied independently on each class. 
Then,~$\AlgUV(\niceJobsetij{p})$ is transformed into a schedule for $\Jobsetij{p}$ by Transformation $\shrinkS$. 
By Observation~\ref{thm:shrink_feasible}, $\Load{\Alg(\Jobsetij{p})}{t} \leq \Load{\AlgUV(\niceJobsetij{p})}{t}$. 
Hence, $\Load{\Alg(\Jobset)}{t} = \sum_p \Load{\Alg(\Jobsetij{p})}{t} \leq \sum_p \Load{\AlgUV(\niceJobsetij{p})}{t}$ for all~$t$.

Consider the timeslot $t$ with peak load in $\Alg(\Jobsetij{})$, $\p(\Alg(\Jobsetij{})) = \Load{\Alg(\Jobsetij{})}{t} \leq\sum_p \Load{\AlgUV(\niceJobsetij{p})}{t} \leq \sum_p \p(\AlgUV(\niceJobsetij{p}))$. 
By Theorem~\ref{thm:UV_peak}, $\p(\Alg(\Jobsetij{})) \leq \sum_p {(6(e+e^2)+3)}\cdot \p(\Opt(\niceJobsetij{p}))$. 

Now we prove that $\p(\Opt(\niceJobsetij{p})) \leq 3 \cdot \p(\Opt(\Jobsetij{}))$. 
Consider the optimal schedule $\Opt(\Jobsetij{p})$, there exists schedule $\Sch(\niceJobsetij{p})$ generated by Transformation $\relaxS$ where $\niceJobsetij{p}$ is the job set corresponding to $\Jobsetij{p}$ generated by $\convert$. 
By Lemma~\ref{thm:relax_load}, the load of any timeslot in $\Sch(\niceJobsetij{p})$ is no more than the sum of loads of three timeslots in $\Opt(\Jobsetij{p})$, and hence no more than three times of the peak in $\Opt(\Jobsetij{p})$. Therefore, $\p(\Opt(\niceJobsetij{p})) \leq \p(\Sch(\niceJobsetij{p})) \leq 3\cdot \p(\Opt(\Jobsetij{p})) \leq 3\cdot \p(\Opt(\Jobsetij{}))$.

In summary, $\p(\Alg(\Jobsetij{}))\leq \sum_p {(6(e+e^2)+3)}\cdot \p(\Opt(\niceJobsetij{p})) \leq \sum_p {(6(e+e^2)+3)}\cdot 3\cdot \p(\Opt(\Jobsetij{})) \leq {(18(e+e^2)+9)}\cdot \lceil\log \frac{\wmax}{\wmin}\rceil \cdot \p(\Opt(\Jobsetij{}))$. 
\end{proof}

Since the $\minimization$ problem is a special case of the $\peak$ problem where jobs have uniform height, we have the following corollary:
\begin{corollary}
The algorithm $\Alg$ is $(18(e+e^2)+9)$-competitive for the $\minimization$ problem.
\end{corollary}

%% file: peak_exact.tex
\subsection{The interval graph approach on the $\peak$ problem}
\label{sec:other_fpt}

In Section~\ref{sec:exact_fpt}, we introduced an exact algorithm $\AlgE$ using the \emph{linear clique arrangement} property of the interval graphs. 
The linear property of the consecutive clique arrangement of interval graphs gives
a direction to design a dynamic programming algorithm, which breaks down a problem
into overlapped subproblems until the subproblems are simple enough to be solved. 
In the following of this section, we show that $\AlgE$ can be used to solve the $\peak$ problem by changing the objective function.

\runtitle{Algorithm $\AlgE_{peak}$ for the $\peak$ problem (also see Section~\ref{sec:exact_fpt}).}
The jobs are considered as time intervals and the time horizon is chopped into ``windows''. 
The algorithm visits all windows accordingly from the left to the right. 
In Stage $i$, the $i$-th window is visited and the algorithm maintains a candidate set of schedules for the visited windows that no optimal solution is deleted from the set. 
In each Stage $i$, the algorithm consists of three procedures: ListConfigurations, ConcatenateTables and FilterTable. 

The ListConfigurations procedure lists all possible configurations (i.e., execution segments) of the jobs in $C_i$ within $\win_i$. The invalid configurations will be deleted. 
The valid configurations together with \textit{the peak load within $\win_i$} will be stored in a table. 

The ConcatenateTables procedure concatenates the configurations in the current window $\win_i$ and the configurations in the windows which have been seen so far. If the execution interval after concatenation is not valid, it is deleted from the table. \textit{The peak load of the new configuration is simply the maximum among the peaks of the two concatenated configurations.} 

The FilterTable procedure filters non-optimal configurations. The idea is, given a configuration of the jobs in $C_i$, there must be a best decision of the jobs in $\bigcup^{i-1}_{k=1} C_k\setminus C_i$ which has minimum \textit{peak load} within the intervals $[0,b_{i+1})$, where $b_{i+1}$ is the right boundary of the window $\win_i$. 
For each configuration, we only keep the (partial) schedule with the minimum \textit{peak load}. 

After processing all the windows, the schedule with minimum \textit{peak load} can be found in the final table.

It can be seen that we list all possible configurations. 
A configuration is deleted only when it is invalid or it is identical to another configuration with lower peak \hhl{peak}.
Hence in the end we get an optimal schedule.
It also shows that the $\peak$ problem is fixed parameter tractable with respect to the maximum width of jobs and the maximum number of overlapped feasible intervals, and the maximum length of windows.


\begin{corollary}
The $\peak$ problem is fixed parameter tractable with respect to the maximum width of jobs, and the maximum number of overlapped feasible intervals.
\end{corollary}

%% file: conclusion.tex
\section{Conclusion}
\label{sec:conclusion}

We develop the first online algorithm with polylog-competitive ratio and the first FPT algorithms for non-preemptive smart grid scheduling problem in general case. We also derive matching lower bound for the competitive ratio. Constant competitive online algorithms are presented for several special input instances. 
\comment{
\hhl{Consider the constant factor $3$ in Corollary~\ref{thm:align_cost}, it is a penalty cost due to round up. If we choose bigger classification factor, it needs to sample more timeslot in order to bound the load after round up (Observation~\ref{obs:aling_overlap}.) Hence, a smaller classification factor brings less penalty cost. However, when the classification factor is between $1$ and $2$, the number of timeslots it needs to bound the load of a timeslot after round up is $2$ or $3$. At the mean time there are more than $\logKw$ classes and it brings even bigger penalty cost. 
By choosing classification factor $2.156$, the competitive ratio for general case is $(35.8584 \logKw)^\alpha \cdot \left( 8 e^\alpha+1\right)$.}
}

There are quite a few directions in extending the problem setting: different cost functions perhaps to capture varying electricity cost over time of the day; jobs with varying power requests during its execution (it is a constant value in this paper); other objectives like response time.
A preliminary result is that we can extend our online algorithm to the case where a job may have varying power requests during its execution,
in other words, a job can be viewed as having rectilinear shape instead of being rectangular.
In such case, the competitive ratio is increased by a factor depending on the maximum height to the minimum height ratio of a job.
Precisely, the competitve ratio becomes
$(36 H\ceilLogKw)^\alpha \cdot (8 e^\alpha+1)$, where
$H = \max_J \frac{h_H(J)}{h_L(J)}$, $h_H(J)$ and $h_L(J)$ are respectively the highest and lowest power request of job $J$.

\comment{
\hhl{
Consider the case that jobs may have different power request on different execution timeslot. For job $J$, denote the highest and lowestt power request as $h_H(J)$ and $h_L(J)$, respectively. Let $H = \max_J \frac{ h_H(J)}{ h_L(J)}$:
\begin{lemma}
\label{lm:rectilinear}
For jobs with rectilinear shape, there is a $(36 H\logKw)^\alpha \cdot \left( 8 e^\alpha+1\right)$-competitive online algorithm.
\end{lemma}
}
}

\comment{
There are many future directions extending from this work. One of them is to obtain an online algorithm with better competitive ratio for general case. Another is to find a tight bound of approximation ratio for the problem. The others includes optimal exact and FPT algorithms. Furthermore, different cost functions for the problem are desired, e.g. electricity costs depend on hours. Also, different types of power requests, such as L-shaped power request, for the problem remain open.
}